\documentclass[conference]{IEEEtran}
\pdfoutput=1 
\usepackage{cite, graphicx,url,setspace,color}
\usepackage{amsmath,amssymb,amsthm}
\usepackage{algorithm}
\usepackage{algorithmic}

\usepackage{amssymb}
\usepackage{bbm}

\usepackage{enumitem}
\usepackage{graphicx} 
\usepackage{caption}
\usepackage{subcaption}
\usepackage{fancybox}
\usepackage{amssymb}
\usepackage{bbm}
\usepackage{bm}
\usepackage{enumitem}
\usepackage{CJK}

\newtheorem{assumption} {Assumption}
\newtheorem{proposition} {Proposition}
\newtheorem{remark}{Remark}

\newtheorem{corollary}{Corollary}
\newtheorem{lemma}{Lemma}
\usepackage{epstopdf}
\newcommand {\ga} {\mathcal{G}} 
\newcommand {\hF} {\hat{F}}
\begin{document}

\title{ Particle Gaussian Mixture (PGM) Filters}

\author{
\IEEEauthorblockN{D. Raihan and S. Chakravorty}
\IEEEauthorblockA{Department of Aerospace Engineering \\ Texas
  A\&M University \\ College Station, TX} 
  }
\maketitle

\setlength{\parskip}{0ex}

\begin{abstract}
 Recursive estimation of nonlinear dynamical systems is an important problem that arises in several engineering applications. Consistent and accurate propagation of uncertainties is important to ensuring good estimation performance. It is well known that the posterior state estimates in nonlinear problems may assume non-Gaussian multimodal densities. In the past, Gaussian mixture filters and  particle filters were introduced to handle non-Gaussianity and nonlinearity. However, these methods have seen only limited success as most mixture filters attempt to fix the number of mixture modes during estimation process, and the particle filters suffer from the curse of dimensionality. In this paper, we propose a particle based Gaussian mixture  filtering approach for the general nonlinear estimation problem that is free of the particle depletion problem inherent to most particle filters. We employ an ensemble of randomly sampled states for the propagation of state probability density. A Gaussian mixture model of the propagated uncertainty is then recovered by clustering the ensemble. The posterior density is  obtained subsequently through a Kalman measurement update of the mixture modes. We prove the weak convergence of the PGM density to the true filter density assuming exponential forgetting of initial conditions by the true filter. The estimation performance of the proposed filtering approach is demonstrated through several test cases. 
\end{abstract}

\section{Introduction}

\IEEEPARstart{R}apid advances in the fields of control and automation has made it necessary to be able to estimate the state of a numerous variety of dynamical systems. As a result, there is growing interest in the recursive and computationally efficient algorithms for estimating the state and associated uncertainty in higher dimensional nonlinear systems. A great deal of prior research is available on static and dynamic estimation of parameters and systems. The Kalman filter was proposed as the unbiased minimum variance estimator for linear dynamical systems perturbed by additive Gaussian noise~\cite{Kal,Bucy}. The extended Kalman filter (EKF) was introduced to incorporate nonlinear systems into the Kalman filtering framework\cite{smith}. However, the limitations of the Jacobian linearization assumptions and the accumulation of linearization errors often resulted in the divergence of EKF estimates. The Unscented Kalman Filter (UKF) and the broader class of sigma point Kalman filters provided a derivative free alternative to the EKF\cite{Jul,Julier,Wan}. It has been shown that the UKF algorithm is in fact  a linearization of the process and measurement functions using statistical regression using the sigma points\cite{Lefeb}. In addition to linearizing the system model, both EKF and UKF approximate the posterior pdf with a single Gaussian pdf. However, the state pdf in a general nonlinear filtering problem can be non-Gaussian and multimodal. A Gaussian mixture approximation of the state pdf was proposed to incorporate the multimodality of the problem in nonlinear settings\cite{Alspach,Sorenson}. These approaches however had a major shortcoming as the number of Gaussian components were kept constant through out the estimation process. Also the component weights were updated only during the measurement update. Approaches to adapting the weights of individual Gaussian modes by minimizing the propagation error committed in the GMM approximation have been proposed recently\cite{Terejanu}. A different approach to improving the accuracy of GMM filters is by splitting the Gaussian components during the propagation based on nonlinearity induced distortion\cite{Mars}. Both of these approaches require frequent optimizations, or entropy calculations, to be performed during the propagation, which significantly add to the overall computational requirement.

The particle filters (PF) are a class of sequential Monte Carlo methods that employ an ensemble of states known as particles to represent the state pdf\cite{Gor,Aru}. These states are sampled from the initial pdf and propagated forward in time based on the nonlinear system model. The measurement updates are performed by assigning weights to individual particles which may then be resampled. The PF does not enforce restrictive assumptions on the nature of dynamics or pdf. Quite often, the measurement updates in particle filters result in weight degeneration wherein a significant fraction of particles lose their importance weights. This problem, termed the ``particle depletion" is a major shortcoming of the particle filters as it requires the number of particles to be increased exponentially with the dimension of state space\cite{Beng,Hua}. Particle based approaches  such as the Ensemble Kalman filter (EnKF) and the Feedback particle filter (FPF) that forego the resampling based measurement update have been demonstrated to be more effective in higher dimensional filtering problems involving unimodal pdfs \cite{Geven,FPF}.\\

In this paper, we propose a particle Gaussian mixture filter (PGMF), addressing the general multimodal nonlinear filtering problem. The PGMF design is inspired by a previous work on a UKF-PF hybrid filter that was proposed for space object tracking\cite{UKPF}. The PGMF employs an ensemble of states for performing the uncertainty propagation. A functional form of the propagated pdf is then recovered as a Gaussian mixture model by clustering the states. The posterior pdf is obtained by performing a Kalman measurement update on the GMM. The PGMF is conceived to keep track of the nonlinear uncertainty propagation without performing any additional optimization and splitting operation during the propagation step. As the posterior pdf is obtained without employing the particle measurement update, the PGMF is not prone to the particle depletion problem and the associated curse of dimensionality. As the additional clustering step is performed only during the measurement update step, the PGMF is especially suitable for filtering in the sparse measurement scenario.

The remainder of this article is organized as follows.
  An introductory discussion on mixture model and clustering is given in section \ref{sec:prem}. The PGM filter algorithm, and an associated convergence result, are presented in section \ref{sec:pgm}. Details pertaining to the actual implementation of the proposed filter is given in section\ref{sec:implem}. The PGM filter is applied to three test cases and compared extensively with the PF, the UKF and the EnKF in section \ref{sec:numex}. 
\section{Preliminaries: Mixture Model Filtering}
\label{sec:prem}
Let the state of the dynamical system of interest be denoted by $x \in \Re^d$. We assume that the state of the system evolves according to a Markov chain whose transition density is specified by $p(x'/x)$, and assumed to be known. We also obtain measurements of the state at discrete times $n$ and the observation model is specified by the following:
\begin{align}
z_n = h(x_n) + v_n,
\end{align}
where $\{v_n\}$ is a discrete time Gaussian white noise process with zero mean and covariance $R_n$. It is very well known that the filtered density of the state of the Markov chain follows the following two steps. Let $\pi_{n-1}(x)$ denote the pdf of the state after the measurement $z_{n-1}$. Then, the prediction of the pdf before the measurement $z_n$ at time $n$ (the predicted prior pdf) is given by:
\begin{align}
\pi_n^-(x) = \int p(x/x')\pi_{n-1}(x')dx', \label{P}
\end{align}
which is the law of total probability. Further, after measurement $z_n$ is received, the pdf of the state is updated according to Bayes rule as (the posterior pdf):
\begin{align}
\pi_n(x) = \frac{p(z_n/x)\pi_n^-(x)}{\int p(z_n/x')\pi_n^-(x')dx'}, \label{U}
\end{align}
$p(z/x)$ is the measurement likelihood function and can be inferred from the measurement model above.
The prediction and the update steps above are the key steps to any recursive filtering algorithm and different filtering approaches are distinguished by how they perform the above two steps. Let us assume that a mixture representation has been chosen for the predicted and posterior pdfs. In particular, let:
\begin{align}
\pi_n^-(x) = \sum_{i=1}^{M^-(n)} \omega_i^-(n) p_{i, n}^-(x), \; \nonumber\\
\pi_n(x) = \sum_{i=1}^{M(n)} \omega_i(n) p_{i,n}(x),
\end{align}
where $p_i^-(.), p_i(.)$ are standard pdfs, and $\{\omega_i^-(n)\}, \{\omega_i(n)\}$ are positive sets of weights that both add up to unity.  The prediction equation for the mixture model then boils down to the following:
\begin{align}
\pi_n^-(x) = \sum_{i=1}^{M(n-1)} \underbrace{\omega_i(n-1)}_{\omega^-_i(n)} \underbrace{\int p(x/x')\pi_{i,n-1}(x')dx'}_{p_{i,n}(x)}.
\end{align}
Explicitly, the mixture prediction step can be split into the following discrete and continuous steps:
\begin{align}
\omega_i^-(n) = \omega_i(n-1), \label{P.D}, \\
p^-_{i,n} (x) = \int p(x/x') p_{i,n-1}(x')dx'. \label{P.C}
\end{align}
Given an observation $z_n$, the prior mixture $\pi_n^-(x)$ is transformed into the posterior mixture $\pi_n(x)$ as follows:
\begin{align}
\pi_n(x) =\frac{ p(z_n/x) \sum_{i=1}^{M^-(n)} \omega_i^-(n) \pi^-_{i,n}(x)}{\int p(z_n/x') \sum_{i=1}^{M^-(n)} \omega_i^-(n) \pi^-_{i,n}(x')dx'} \nonumber\\
= \frac{\sum_{i=1}^{M^-(n)} \omega_i^-(n) p(z_n/x)\pi^-_{i,n}(x)}{\sum_{i=1}^{M^-(n)} \omega_i^-(n) \int p(z_n/x')\pi^-_{i,n}(x')dx'}.
\end{align}
Define the likelihood that $z_n$ comes from the $i^{th}$ mixture component as:
\begin{align}
l_i(n) \equiv \int p(z_n/x')\pi_{i,n}^-(x')dx'.
\end{align}
Rearranging the above mixture expression using the definition of the component/ mode likelihood gives us:
\begin{align}
\pi_n(x) = \sum_{i=1}^{M^-(n)} \underbrace{\frac{w_i^-(n)l_i(n)}{\sum_j w_j^-(n)l_j(n)}}_{w_i(n)} \underbrace{\frac{p(z_n/ x) \pi_{i,n}^-(x)}{l_i(n)}}_{\pi_{i,n}(x)}.
\end{align}
The above expression clearly shows that the measurement update has a hybrid nature, a standard update of the individual modes of the mixture with the measurement $z_n$, and a discrete Bayesian update of the mode weights using the mode likelihoods $l_i(n)$. Note further that the mode likelihoods are the Bayes' normalization factors for the individual modes. Moreover, note that there are no approximations whatsoever in the above mixture update equations. Explicitly, we delineate the discrete and continuous updates of the mixture model below:
\begin{align}
\omega_i(n) = \frac{w_i^-(n)l_i(n)}{\sum_j w_j^-(n)l_j(n)}, \label{MU.D}, \\
p_{i,n}(x) = \frac{p(z_n/ x) p_{i,n}^-(x)}{l_i(n)}. \label{MU.C}
\end{align}
It behooves us to take a closer look at the hybrid prediction equations \ref{P.D} and \ref{P.C} as well as the hybrid update equations \ref{MU.D} and \ref{MU.C}. It has been shown that a Gaussian mixture model(GMM) can be used to approximate a general continuous non-Gaussian pdf to any degree of accuracy\cite{Alspach}. The GMMs also inherit many of the desirable properties of Gaussian densities which makes their analysis easier. In figure \ref{fig:gmmbasic}, a weighted Gaussian sum is used to represent a multimodal pdf. Let us assume that we have fixed the form of the mixture model to a Gaussian Mixture Model (GMM), i.e., the posterior pdf at time $n-1$ can be represented by the GMM:
\begin{align}
\pi_{i,n-1}(x) = \ga (x;\mu_i(n-1), P_i(n-1)),
\end{align}
where $\ga(x; \mu, P)$ represents the Gaussian pdf with mean $\mu$ and covariance $P$.
Consider first the prediction equations. Note that from the way it has been written, the number of mixture components at time $n-1$,  $M(n-1)$, is the same as the number of mixture components of the prediction at time $n$, $M^-(n)$. However, this assumes that the prediction of the $i^{th}$Gaussian component $\pi_{i, n-1}$ of the posterior pdf at time $n-1$ remains a single Gaussian at time $n$, $\pi_{i,n}^-$. However, this is, in general, not true. The number of mixture components necessary to approximate the state pdf may vary from one time step to the other. For example, consider the nonlinear dynamical system given by
\begin{align}
\begin{bmatrix}
\dot{x}_{1}\\ 
\dot{x}_{2}
\end{bmatrix}
=&\begin{bmatrix}-\frac{x_{1}}{2}\\\sin(\frac{x_{2}}{2})\end{bmatrix}
+\Gamma(t),\\
\pi_{0}(X)=&\ga(X,\begin{bmatrix}-12\\0 \end{bmatrix},\begin{bmatrix}0.2&0\\0&1\end{bmatrix}).
\end{align}

In figure \ref{fig:multiprop}, the locations of 200 particles sampled from the initial pdf $\pi_{0}(x)$ as they evolve through the dynamics of the system are plotted. The particles are seen to separate into two distinct modes as time progresses.   
  Hence, in order to use mixture models for prediction, we have to find a way to tackle the problem of time varying number of GMM components. \\
\begin{figure}[h]     
\includegraphics[width=0.45\textwidth,height=0.2\textheight]{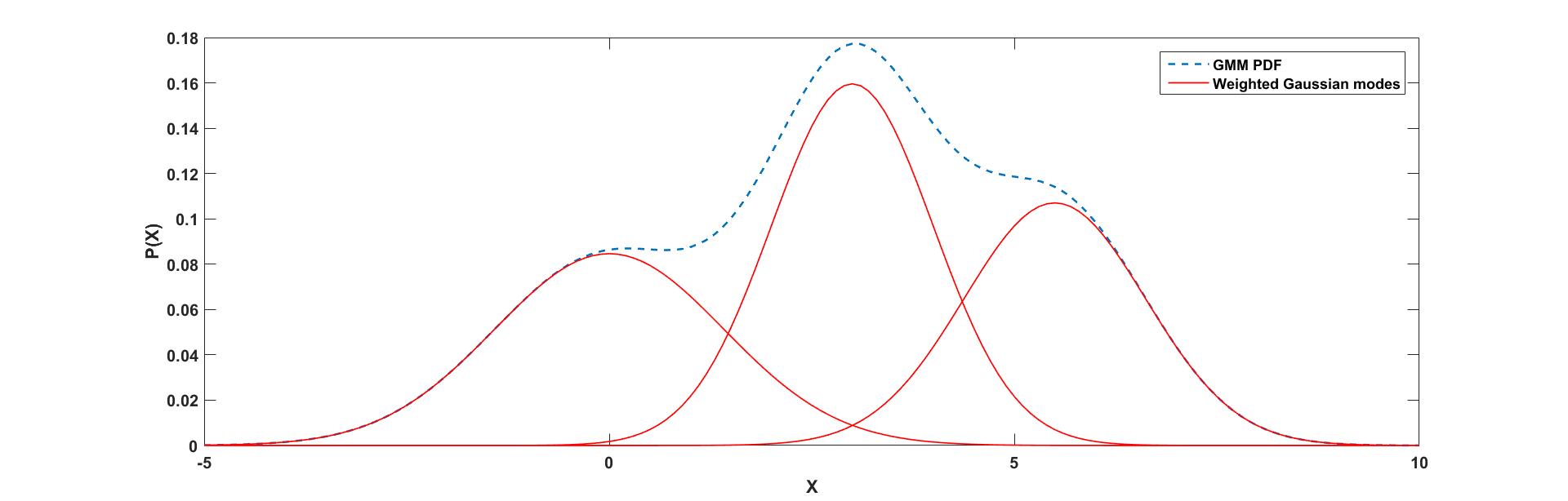} 
\caption{Gaussian Mixture representation of a multimodal PDF}
\label{fig:gmmbasic}
\end{figure}
Next, let us consider the measurement update equations \ref{MU.D} and \ref{MU.C}.
The discrete weight update Eq. \ref{MU.D} is quite clear, however, it behooves us to take a closer look at the continuous mode update Eq. \ref{MU.C}. Since the prior component is Gaussian, and the update Eq. \ref{MU.C} could be done simply using the Kalman/ least squares update, i.e.,:
\begin{align}
\mu_i(n) = \mu_i^-(n) + P_{i, zx}^{-1'}(n) P_{i,zz}^{-1}(n) (z_n - E_i[h(X)]), \label{KU1} \\
P_i(n) = P_i^-(n) - P_{i, zx}^{-1'}(n) P_{i,zz}^{-1}(n)P_{i, zx}^{-1}(n), \label{KU2}
\end{align} 
where $P_{i,zx}(n) = E_i[h(X)-E_{i}(h(X))(X-E_{i}(X))']$ ,$ P_{i,zz}(n) = E_i[(h(X)-E_{i}(h(X))(h(X)-E_{i}(h(X))']$, and $E_i[f(X)]$ represents an expectation of the function $f(X)$ with respect to the random variable $X$ where $X \sim\ga(x; \mu_i^-(n), P_i^-(n))$.  However, this update is not necessarily correct. Similar to the prediction case, in general, a single predicted Gaussian component can split into multiple modes after the Bayesian update \ref{MU.C}, i.e., the updated component itself is a GMM.
An illustration of this is given in figure. \ref{fig:multiup}. In this case we have a prior ensemble generated from $\pi(x)=\ga(X,\begin{bmatrix}0\\0 \end{bmatrix},\begin{bmatrix}1&0\\0&2\end{bmatrix})$. Then, a noisy measurement $z=2$ is recorded where
\begin{align}
z=x_{1}^{2}+\tau,\\
\tau \sim \ga(x,0,2)\nonumber.
\end{align}
An ensemle for the posterior pdf $\pi(x|z)$ is obtained through resampling and is seen to split into two separate modes. Hence, just as in the prediction step, there is a need to deal with the time varying number of GMM components after an update.\\

\begin{figure}[h!]    
\begin{subfigure}[t]{0.45\textwidth}
\includegraphics[width=\linewidth,height=4cm]{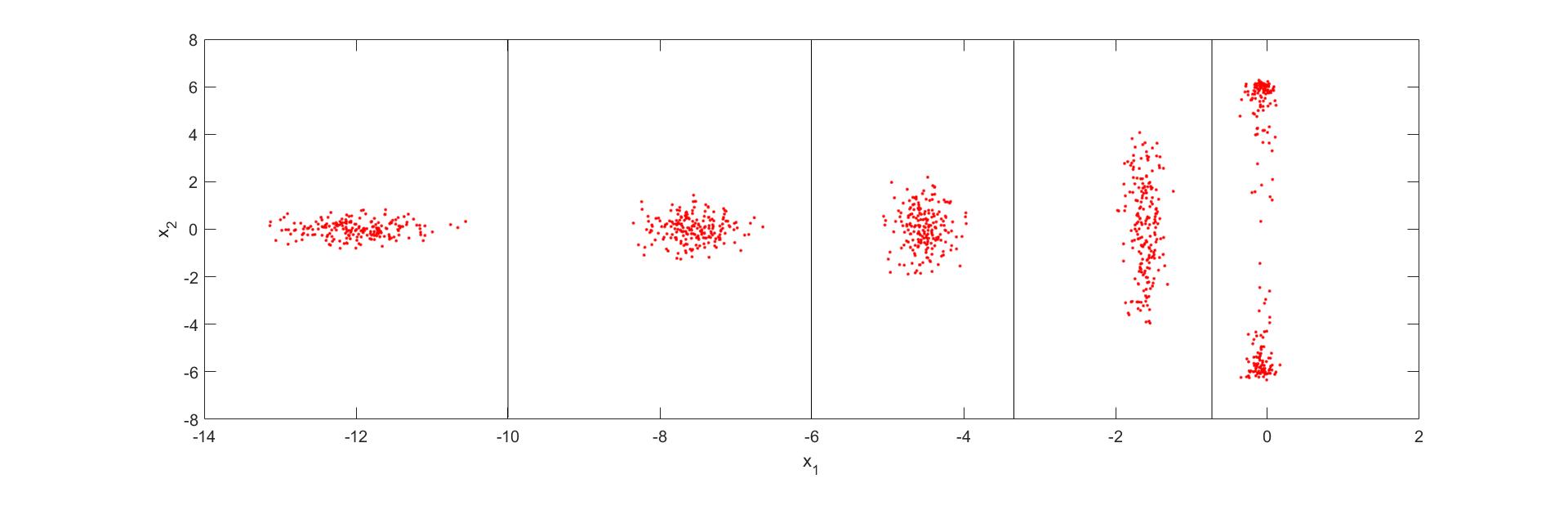}
\caption{}
\label{fig:multiprop}
\end{subfigure}
\hspace{\fill}
\begin{subfigure}[t]{0.45\textwidth}
\includegraphics[width=\linewidth,height=4cm,]{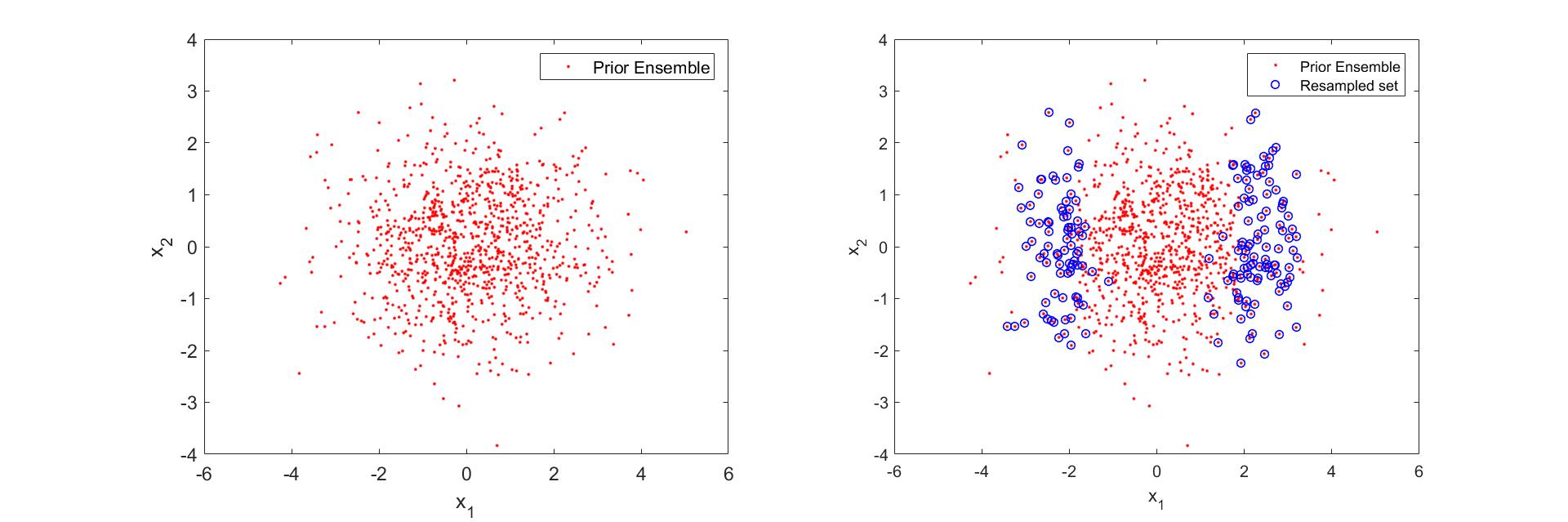}
\caption{}
\label{fig:multiup}
\end{subfigure}
\caption{ Formation of multimodality a) through dynamics b) through measurement update}
\end{figure}

For ease of treatment and clarity of exposition, we shall not consider the measurement update aspect of the GMM filtering problem in this paper, which will be treated in a companion paper. Hence, we make the following assumption for the remainder of the paper.
\begin{assumption} \label{UM_update}
We shall assume a Gaussian mixture representation for the predicted and posterior filtered densities. Further, we assume that given a predicted mixture component at time $n$, $\ga(x; \mu_i^-(n), P_i^-(n))$, the update Eq. \ref{MU.C} after an observation $z_n$ is approximated arbitrarily well by the Least Squares/ Kalman update Eq. \ref{KU1}-\ref{KU2}.
\end{assumption}

\section{The Particle Gaussian Mixture (PGM) Filter}
\label{sec:pgm}
In this section, we first present the PGM filter. In the following subsection, we show the weak convergence of the PGM filter density to the true filter density under the condition of exponential forgetting of initial conditions by the true filter. Finally, we also compare and contrast the PGM algorithm with other mixture based nonlinear filters, in particular, the PF, the Gaussian Mixture Filter (GMF) and the EnKF.

\subsection{The PGM Algorithm}
The basic assumption underlying the PGM algorithm is that the predicted prior and posterior filter densities can be represented using a GMM. In particular, let:
\begin{align}
\pi_n^-(x) = \sum_{i=1}^{M^-(n)} \omega_i^-(n) \ga_i^-(x; \mu_i^-(n), P_i^-(n)),\\
\pi_n(x) = \sum_{i=1}^{M(n)} \omega_i(n) \ga_i(x; \mu_i(n), P_i(n)).
\end{align}
In general, $M^-(n)$  and $M(n)$ need not be the same, however, owing to Assumption \ref{UM_update}, they are assumed to be equal for the purposes of this paper. For instance, given a linear measurement function, this is true. The PGM filtering algorithm is composed of three basic steps that are described below.
\begin{enumerate}
\item \textbf{Sampling}:
 The PGM filter assumes the availability of the Markov transition kernel $p(x/x')$ using which it can draw samples of the next state $x$ given that the current state is $x'$. 
\begin{algorithm}
\caption{PGM Algorithm}
Given $\pi_0 (x) = \sum_{i=1}^{M(0)} \omega_i(0) \ga_i(x; \mu_i(0), P_i(0))$, transition density kernel $p(x'/x)$, n = 1. 
\begin{enumerate}
\item{Sample $N_{p}$ particles $X^{(i)}$ from from $\pi_{n-1}$ and the transition kernel $p(x'/x)$ as follows:}
\begin{enumerate}
\item{Sample $X^{(i)'}$ from $\pi_{n-1}(.)$.}
\item{Sample $X^{(i)}$ from $p(./X^{(i)'})$.}
\end{enumerate}
\item{Use a Clustering Algorithm $\mathcal{C}$ to cluster the set of particles $\{X^{(i)}\}$ into $M^-(n)$ Gaussian clusters with weights, mean and covariance given by $\{w_i^-(n), \mu_i^-(n), P_i^-(n)\}$.}
\item{Update the mixture weights and the mixture means and covariances to $\{\omega_i(n), \mu_i(n), P_i(n)\}$, given the observation $z_n$, utilizing the Kalman update Eqs. \ref{KU1}, \ref{KU2}.}
\item{n = n+1, go to Step 1.}
\end{enumerate}
\end{algorithm}
The first step in the PGM algorithm is the use of the transition kernel to generate a set of samples at the next time step (which is the same as in a Particle filter).
\item \textbf{Clustering}: Then, we use a clustering algorithm $\mathcal{C}$ to partition the set of points into $M^-(n)$ different clusters whose means and covariances can be evaluated using sample averaging.
Clustering is a field of Machine learning termed as Unsupervised Learning \cite{Duda,Jain1}. There are many different algorithms that can be used for clustering, for instance K-means clustering \cite{Macq}, EM clustering \cite{Demp} among others. In the basic clustering algorithms, the number of clusters have to be specified\cite{sugar,tibsh}  while the more advanced techniques also estimate the number of clusters \cite{figjain}. In the experimental results presented in this paper, we use the simple K-means clustering algorithm, which is computationally very inexpensive while still being able to give good results for well separated clusters.
The k-means clustering is a popular approach to partitioning wherein the data set is grouped into different clusters so that the sum of squares of within-group distances is minimized,i.e, the data set S is partitioned into $M$ clusters $G_{M}^{*}=\lbrace S_{1},\cdots,S_{M} \rbrace $ such that
\begin{equation}
G_{L}^{*}= \underset{G_{L}}{\operatorname{arg min}} \sum \limits_{{i=1}}^{M} \sum \limits_{x_{j} \in S_{i}}\|x_{i}-\mu_{i}\|^{2}.
\end{equation}
Here $G_{M}$ denotes any partition of the set S into $M$ clusters and $\mu_{i}$ represents the mean of the elements of the $i^{th}$ cluster in that partition. Once the vectors $x_{i}$ are assigned into different clusters, an $M$ mode GMM describing the set S may be derived as follows.
\begin{align}
n_{i}= &\sum \limits_{j=1}^N \mathbbm{1}(x_{j} \in S_{i} ),\\
w_{i}= & \frac{1}{N}n_{i},\\
\mu_{i}=&\frac{1}{n_{i}}\sum \limits_{x_{j} \in S_{i}} x_{j},\\
C_{i}=&\frac{1}{n_{i}-1}\sum \limits_{x_{j} \in S_{i}} ( x_{j}-\mu_{i})( x_{j}-\mu_{i})^{T}.
\end{align}
Here $\mathbbm{1}(.)$ represents the indicator function. In figure \ref{fig_sim}, a set of points are chosen from the two dimensional space and clustered using the k-means algorithm to obtain a GMM.

\begin{figure}[h]
\centering
\begin{subfigure}[t]{0.45\textwidth}
\includegraphics[width=\linewidth,height=4cm,]{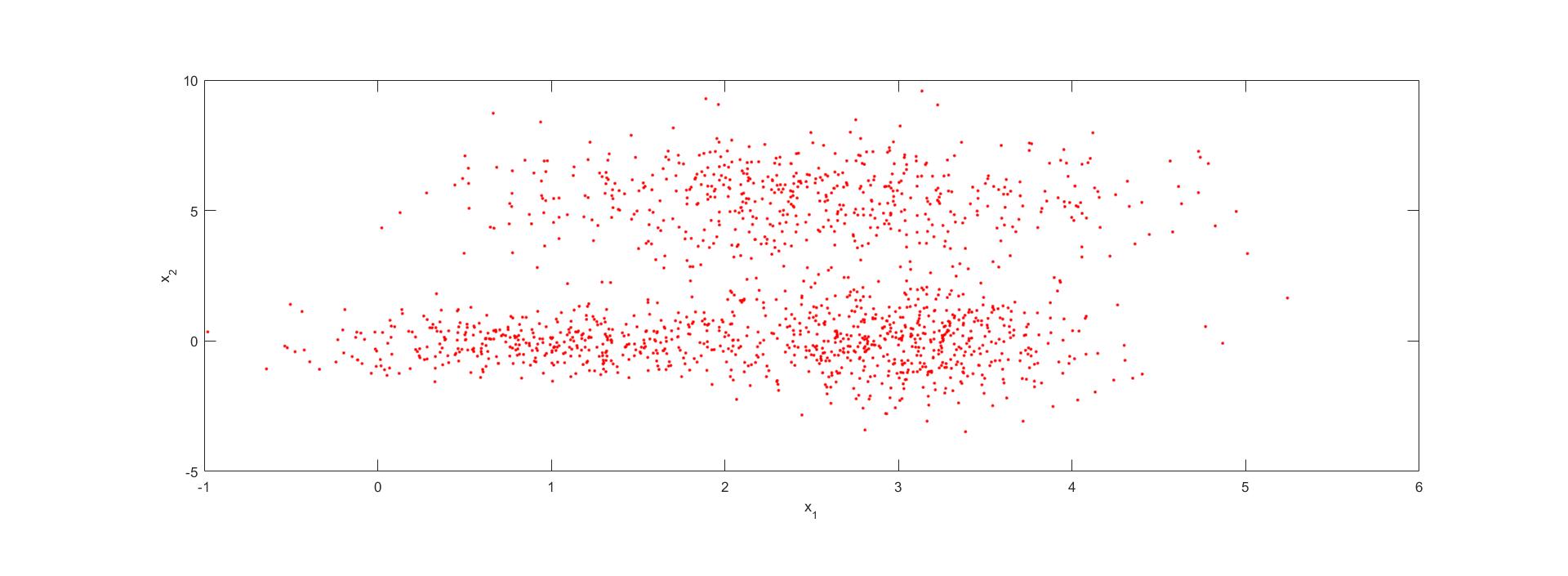}%
\label{fig_first_case}
\end{subfigure}
\begin{subfigure}[t]{0.45\textwidth}
\includegraphics[width=\linewidth,height=4cm,]{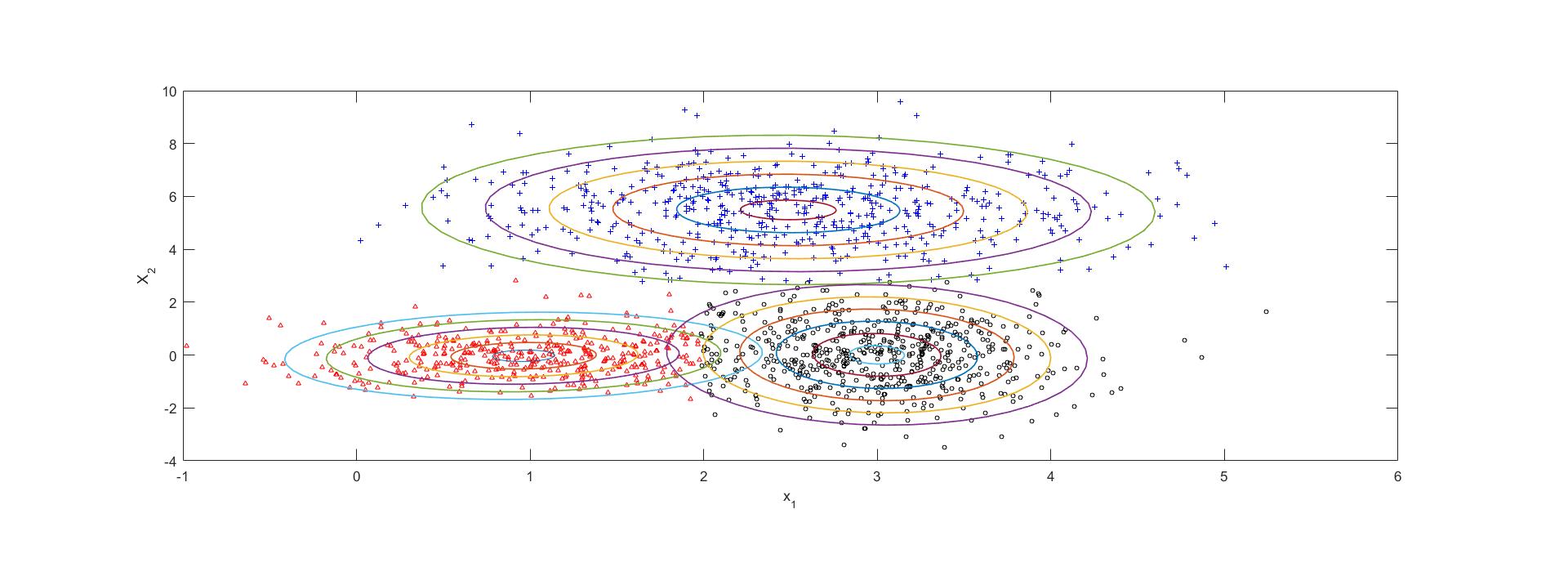}%
\label{fig_second_case}
\end{subfigure}
\caption{(a) Set of points in the 2-d Euclidean space (b) Contour plots of the GMM components obtained by clustering the points into three groups using k-means algorithm}
\label{fig_sim}
\end{figure}
\item \textbf{Measurement update}: Incorporate the measurement information by updating the means and covariances of all $M$ modes individually using a least squares/ Kalman measurement update. Also update the mixture weights using the mode likelihoods $l_{i}(n)$ as in  Eq. \ref{MU.D}.
\end{enumerate}
\subsection{Analysis of the PGM Filter}
In the following, we analyze the PGM filter. We show that under the assumption of a perfect clustering scheme $\mathcal{C}$, the PGM filter density converges weakly to the true filter density.\\
Let $F_{z_n}(\pi_{n-1}) = \pi_n$ denote the true filter density at time $n$ given that the filter density at time $n-1$ is $\pi_{n-1}$ and the observation at time $n$ is $z_n$. Further, let $\hF_{z_n}(\pi_{n-1})$ denote the filter density approximated by the PGM filter. We make the following exponential forgetting assumption on the true filter.
\begin{assumption} \label{EF}
We assume that there exists $C < \infty$ and $\rho<1$ such that:
\begin{align}
||F_{z_n}F_{z_{n-1}} ..F_{z_1}(\pi_0) - F_{z_n}F_{z_{n-1}} ..F_{z_1} (\pi_0')|| \leq C\rho^n ||\pi_0 -{\pi'}_0||,
\end{align}
for any measurement sequence $\{z_1,z_2,\cdots z_n\}$, any $\pi_0, \pi_0'$, and where $||.||$ denotes the $L_1$ norm.
\end{assumption}
Similarly let $\hF_{z_n}\hF_{z_{n-1}}\cdots\hF_{z_0}(\pi_0)$ denote the filtered density approximated by the PGM filter given the measurement sequence $\{z_1, z_2, \cdots z_n\}$ and the initial density $\pi_0$. 
\begin{assumption} \label{S_acc}
Let $Prob(||\hF_{z_n} (\hat{\pi}_{n-1}) - F_{z_n}(\hat{\pi}_{n-1})||> \epsilon) < \delta$, for all $n$. Further, we asume that $Prob( ||\hF_{z_n} (\hat{\pi}_{n-1}) - F_{z_n}(\hat{\pi}_{n-1})||> M) = 0$, for all $n$, for some $M<\infty$ (the error in a one step approximation of the filter density is almost surely uniformly bounded over all time).
\end{assumption}
\begin{lemma}\label{L1}
Let $||\hF_{z_n}(\hat{\pi}_{n-1}) - F_{z_n}(\hat{\pi}_{n-1})|| \leq \epsilon$, for all $n$. Under Assumption \ref{EF}, it follows that $||\hat{\pi}_n - \pi_n|| \leq \frac{C\epsilon}{1-\rho}$.
\end{lemma}
\begin{proof}
We have:
\begin{align}
\hat{\pi}_n - \pi_n = \hF_{z_n}\hF_{z_{n-1}}..\hF_{z_1} (\pi_0) - F_{z_n}F_{z_{n-1}}...F_{z_1}(\pi_0), \nonumber\\
= \underbrace{[\hF_n\hF_{n-1}..\hF_1(\pi_0) - F_n\hF_{n-1}..\hF_1(\pi_0)]}_{\Delta_n}  \nonumber\\
+\underbrace{[F_n\hF_{n-1}..\hF_1(\pi_0) - F_nF_{n-1}\hF_{n-2}..\hF_1(\pi_0)]}_{\Delta_{n-1}} \nonumber\\
+\cdots +\underbrace{[F_nF_{n-1}....\hF_1(\pi_0) - F_n..F_1(\pi_0)]}_{\Delta_1}.
\end{align}
Note that the different terms on the RHS above are:
\begin{align}
\Delta_n = \hF_n(\hat{\pi}_{n-1}) - F_n(\hat{\pi}_{n-1}), \nonumber\\
\Delta_{n-1} = F_n(\hF_{n-1}(\hat{\pi}_{n-2}))- F_n(F_{n-1} (\hat{\pi}_{n-2})), ....\nonumber\\
\Delta_1 = F_n..F_2(\hF_1(\pi_0)) - F_n..F_2(F_1(\pi_0)).
\end{align}
Using Assumption \ref{EF} and the fact that $||\hF_{z_n}(\hat{\pi}_{n-1}) - F_{z_n}(\hat{\pi}_{n-1})|| \leq \epsilon$, for all $n$, it follows that:
\begin{align}
||\Delta_i|| \leq C\rho^{n-i}\epsilon,\nonumber
\end{align}
and thus, 
\begin{align}
||\hat{\pi}_n - \pi_n|| \leq \sum_{i=1}^n  C\rho^{n-i}\epsilon \leq \frac{C\epsilon}{1-\rho}.
\end{align}
The above result also holds for initial conditions in the infinite past, i.e., at $n =-\infty$. In the following, we assume that the initial condition was in the infinite past.
\end{proof}
\begin{lemma} \label{L2}
Let Assumptions \ref{EF} and \ref{S_acc} hold. Given any $\delta, \nu > 0$, there exists an $N < \infty$, such that:
\begin{align}
Prob(||\hat{\pi}_n - \pi_n|| > \frac{(1+\nu)C\epsilon}{1-\rho} )\leq N\delta,
\end{align}
where $ N = n-n'$, and $n'$ is such that $\sum_{i=-\infty}^{n'} \rho^{n-i} = f$, and $f = \frac{\nu C\epsilon}{M(1-\rho)}$.
\end{lemma}
\begin{proof}
Let $e_n = ||\hat{\pi}_n - \pi_n||$, and let $\epsilon_k = ||\hF_{z_k}(\hat{\pi}_{k-1}) - F_{z_k}(\hat{\pi_{k-1}})||$. It follows that $e_n = \sum_{k=-\infty}^n \rho^{n-k}\epsilon_k$. Choose $n'$ such that 
$\sum_{i=-\infty}^{n'} \rho^{n-i} \geq f$, where $f =\frac{\nu C\epsilon}{M(1-\rho)}$. Then:
\begin{align}
e_ n = \underbrace{\sum_{k=n'}^n \rho^{n-k} \epsilon_k}_{\bar{e}_n} + \underbrace{\sum_{k=-\infty}^{n'} \rho^{n-k}\epsilon_k}_{\Delta_n}.
\end{align}
From Assumption \ref{S_acc}, it follows that $Prob(||\Delta_n|| > fM) = 0$, and thus, 
\begin{align} \label{le1}
Prob(||\Delta_n|| > \frac{\nu C \epsilon}{1-\rho}) = 0.
\end{align}
 Similarly, from Lemma \ref{L1}, it follows that:
\begin{align}\label{le2}
Prob(\bar{e}_n > \frac{C\epsilon}{1-\rho}) \leq (n-n')\delta \equiv N\delta.
\end{align}
Using the equations \ref{le1} and \ref{le2}, it follows that $Prob(e_n > \frac{(1+\nu)C\epsilon}{1-\rho}) \leq N\delta$.
\end{proof}
\begin{lemma} \label{L3}
Given a random variable $X$, and a function $f(X)$ such that $E(f(X)) = \bar{f}$, and $Var(f(X)) = \sigma_f^2$, let $\hat{f} = \frac{1}{N} \sum_{i=1}^N f(X_i)$, where $X_i$ are samples of the r.v. $X$. For large enough $N$, it follows that:
\begin{align}
Prob(|\hat{f} - \bar{f} | > \epsilon) < \frac{\sqrt{N}}{\sqrt{2}\pi\sigma_f}e^{-\frac{\epsilon^2 N}{2\sigma_f^2}}.\nonumber
\end{align}
\end{lemma}
\begin{proof}
It can be shown easily that $E(\hat{f}) = \bar{f}$, and $Var(\hat{f}) = \frac{\sigma_f^2}{N}$. Application of the Central Limit Theorem implies that $(\hat{f} - \bar{f})$ is Normal distributed with  mean $E(\hat{f}) - \bar{f} = 0,$ and variance $Var(\hat{f})$. The result then immediately follows.
\end{proof}
The two results above establish that if the sampling error at each step in the filter is small enough, and under the condition of exponential forgetting of initial conditions, then the true filter density can be approximated arbitrarily closely with arbitrary high confidence. In the following, we establish that the sampling error at each step in the PGM filtering process can be arbitrarily small and thus, it follows from the two results above that the PGM filter can approximate the true filter density with arbitrarily high accuracy and arbitrarily high confidence. First,we define the following:
\begin{align}
P(\hat{\pi}_{n-1}) \equiv \hat{\pi}_n^- = \sum_{i=1}^{M^-(n)} \hat{\omega}_i^-(n) \ga_i(x; \hat{\mu}_i^-(n), \hat{P}_i^-(n)), \\
\hat{P}(\hat{\pi}_{n-1}) \equiv \hat{\hat{\pi}}_n^- = \sum_{i=1}^{M^-(n)} \hat{\hat{\omega}}_i^-(n) \ga(x; \hat{\hat{\mu}}_i^-(n), \hat{\hat{P}}_i^-(n)),\\
F_{z_n}(\hat{\pi}_{n-1}) = \sum_{i=1}^{M(n)} \hat{\omega}_i(n)\ga(x; \hat{\mu}_i(n), \hat{P}_i(n)), \\
\hat{F}_{z_n}(\hat{\pi}_{n-1}) = \sum_{i=1}^{M(n)}\hat{\hat{\omega}}_i(n) \ga(x; \hat{\hat{\mu}}_i(n), \hat{\hat{P}}_i(n)).
\end{align}
The above represent the true and the approximate PGM predicted and filtered densities at time $n$ given the approximate density $\hat{\pi}_{n-1}$ at time $n-1$. We have the following result:
\begin{lemma} \label{L4}
Given the GMM representation of the prior pdf above, and a perfect Clustering algorithm $\mathcal{C}$, given any $\epsilon' >0$, andy $\delta'>0$, there exists an $N_{\epsilon',\delta'}(n) < \infty$ such that: if the number of samples used to approximate the predicted pdf at time $n$ is greater than $N_{\epsilon', \delta'}(n)$ then:
\begin{align}
Prob(|\hat{\hat{\omega}}_i^-(n) - \hat{\omega}_i^-(n) > \epsilon') < \delta', \\
Prob(|\hat{\hat{\mu}}_i^{j-}(n) - \hat{\mu}_i^{j-}(n)|> \epsilon') < \delta', \\
Prob(|\hat{\hat{P}}_i^{jk-}(n) - \hat{P}_i^{jk-}(n)|> \epsilon') < \delta',
\end{align}
for all $i, j, k$, where $\hat{\mu}_i^{j-}$ represents the $j^{th}$ element of the mean vector $\hat{\mu}_i^-$ and $\hat{P}_i^{jk-}$ represents the $(j,k)^{th}$  element of the covariance matrix $\hat{P}_i^-$. 
\end{lemma}
\begin{proof}
Given a sample $X^-_i$, the clustering algorithm assigns the sample to a particular cluster $C_j$. The cluster $C_j$ is the correct one under the assumption of a perfect Clustering algorithm. Let $N_j$ denote the number of samples assigned to $C_j$ among all the $N$ samples. Further, let us assume that $N$ is large enough such that $\frac{N_j}{N} \approx \hat{\omega}_j^-$ (we drop the explicit reference to the time $n$ for notational convenience in the following). This assumption can be relaxed in a straightforward fashion but at the cost of a more tedious development which we forego here for the sake of clarity.\\
Consider the cluster $C_j$: its mean and covariance $\hat{\mu}_j^-, \hat{P}_j^-$ are calculated from the $N_j$ samples. Estimates of all the elements of the mean and covariance functions are ensemble averages of particular functions of the random variable $X^-(n)$, i.e, the predicted state of the system. Let $f_m(X^-), m=1,2\cdots M$ denote the all the functions corresponding to the various elements in the mean and covariance of $X^-$. Choose:
\begin{align}
\bar{\sigma}_j^2 = \max_m \sigma_{m,j}^2, \mbox{where}\, \sigma_{m,j}^2 = Var(f_m(X^-)). \label{S1}
\end{align}
If we choose $N>N^j_{\epsilon',\delta'}$ such that
\begin{align}
\frac{\sqrt{N}}{\sqrt{2\pi}\bar{\sigma}_j}e^{-\frac{\epsilon^2 N}{2\bar{\sigma}_j^2}} < \delta', \label{S2}
\end{align}
then it follows from Lemma \ref{L3} that $Prob(|\hat{f}_m -\bar{f}_{i,m}|> \epsilon') < \delta'$, for all $m$, where $\hat{f}_{i,m}$ is the sample average of the parameters $\bar{f}_m$ corresponding to the different elements of the mean and covariance. Such $N_{\epsilon',\delta'}^j$ can be found for all clusters $C_j$ and given that we choose $N_{\epsilon', \delta'}$ as:
\begin{align}
N_{\epsilon', \delta'} = \max_j \frac{N^j_{\epsilon', \delta'}}{\hat{\omega}_j^-}, \label{S3},
\end{align}
 it is guaranteed that all the elements of the mean vector $\hat{\mu}^-(n)$ and the covariance matrix $\hat{P}^-(n)$ can be estimated to an accuracy of $\epsilon'$ with confidence of at least $1-\delta'$, which completes the proof of the result. 
\end{proof}
It may be shown that under Assumption \ref{UM_update} that the Kalman update is an arbitrarily accurate approximation of the true update, the error incurred in estimating the posterior mean and covariance $\hat{\mu}_i(n), \hat{P}_i(n)$ is at most $K(n)\epsilon'$, for some time varying $K(n)< \infty$ which depends on the posterior mean and covariance, given that the predicted prior means and covariances of the clusters of the GMM have been approximated to an accuracy of $\epsilon'$. This can be summarized in the following result:
\begin{lemma} \label{L5}
Given any $\epsilon', \delta'>0$, choose $N_{\epsilon',\delta'}(n)$ according Eqs. \ref{S1}-\ref{S3}. If the number of samples used in the PGM filter to approximate the predicted prior pdf at time $n$ is greater than $N_{\epsilon',\delta'}(n)$ then, there exists $k(n)< \infty$ s.t:
\begin{align}
Prob(|\hat{\hat{\omega}}_i(n) - \hat{\omega}_i(n) > K(n)\epsilon') < \delta', \\
Prob(|\hat{\hat{\mu}}_i^{j}(n) - \hat{\mu}_i^{j}(n)|> K(n)\epsilon') < \delta', \\
Prob(|\hat{\hat{P}}_i^{jk}(n) - \hat{P}_i^{jk}(n)|> K(n)\epsilon') < \delta',
\end{align}
for all $i,j, k$.
\end{lemma}
Next, we find a bound on the $L_1$ error between the estimated and true filtered densities given the error between the parameters of the GMM representing the true and the approximate filtered densities. 
\begin{lemma}\label{L6}
Let $|\hat{\hat{\omega}}_i(n) - \hat{\omega}_i(n)|< \epsilon'$, $|\hat{\hat{\mu}}_i^{j}(n) - \hat{\mu}_i^{j}(n)|< \epsilon'$, and $|\hat{\hat{P}}_i^{jk}(n) - \hat{P}_i^{jk}(n)|< \epsilon$ for all $i,j,k$. Then , given that the state of the system $x \in \Re^d$, there exists $C(n) < \infty$ such that $||\hat{\hat{\pi}}_n - \hat{\pi}_n|| < C(n)d\epsilon'$.
\end{lemma}
\begin{proof}
We show the result for the case of a simple one component Gaussian with an error in the covariance, it can be generalized to the GMM in a relatively straightforward fashion but at the expense of a very tedious derivation which we forego here for clarity. We also suppress the explicit dependence on time $n$ in the following for notational convenience. \\
\begin{align}\label{app1}
\hat{\hat{\pi}} (x) -\hat{\pi}(x) = \frac{1}{(2\pi)^{d/2}|\hat{\hat{P}}|^{1/2}}e^{-\frac{1}{2} (x-\mu)'\hat{\hat{P}}^{-1}(x-\mu)} - \nonumber\\
\frac{1}{(2\pi)^{d/2}|{\hat{P}}|^{1/2}}e^{-\frac{1}{2} (x-\mu)'{\hat{P}}^{-1}(x-\mu)}, \nonumber\\
\approx \frac{1}{(2\pi)^{d/2}|{\hat{P}}|^{1/2}} e^{-\frac{1}{2} (x-\mu)'{\hat{P}}^{-1}(x-\mu)}\nonumber\\
\times \frac{1}{2}(x-\mu)'(\hat{P}^{-1}\Delta \hat{P}^{-1})(x-\mu),
\end{align}
where $\hat{\hat{P}} = \hat{P} + \Delta$ and since 
\begin{align}
e^{-\frac{1}{2} (x-\mu)'{(\hat{P} + \Delta)}^{-1}(x-\mu)} \nonumber\\
\approx e^{-\frac{1}{2} (x-\mu)'\hat{P} ^{-1}(x-\mu)} e^{-\frac{1}{2} (x-\mu)'\hat{P}^{-1} \Delta\hat{P}^{-1}(x-\mu)},
\end{align}
which in turn implies Eq. \ref{app1}. This in turn implies that:
\begin{align}
||\hat{\hat{\pi}} -\hat{\pi}|| \approx \frac{1}{(2\pi)^{d/2} |\hat{P}|^{1/2}} \times\nonumber\\
\int e^{-\frac{1}{2} (x-\mu)'\hat{P}^{-1}(x-\mu)}\frac{1}{2} (x-\mu)'\hat{P}^{-1}\Delta \hat{P}^{-1}(x-\mu)dx, \nonumber\\
\leq \frac{C(\hat{P})\epsilon'}{(2\pi)^{d/2} |\hat{P}|^{1/2}}\nonumber\\
\times \int e^{-\frac{1}{2} (x-\mu)'\hat{P}^{-1}(x-\mu)}\frac{1}{2} (x-\mu)'\hat{P}^{-1}(x-\mu)dx,
\end{align}
since there exists $C(\hat{P}) < \infty$ such that $\frac{1}{2} (x-\mu)'\hat{P}^{-1}\Delta \hat{P}^{-1}(x-\mu) \leq C(\hat{P}) \epsilon' \frac{1}{2} (x-\mu)'\hat{P}^{-1}(x-\mu)$, owing to the fact that $||\Delta|| < L \epsilon'$ for any suitable norm on the perturbation $\Delta$.  \\
Now, let $Y= \hat{P}^{-1/2}(X-\mu)$. Then, it follows that:
\begin{align}
||\hat{\hat{\pi}} -\hat{\pi}|| \leq C(\hat{P})\epsilon' \frac{1}{(2\pi)^{d/2}} \int e^{-1/2 y'y}y'y dy = C(\hat{P})\epsilon' d.
\end{align}
The last step in the above equation follows from noting that $Y'Y$ is a chi-squared random variable of degree of freedom $d$ and thus, its expected value is $d$. This establishes our result. In general for a GMM, the constant $C(n)$ would depend on the means and covariances of all the GMM components and their weights.
\end{proof}
Lemma \ref{L5} and \ref{L6} immediately lead us to the following corollary.
\begin{corollary} \label{C7}
Let $\epsilon'(n)$ be the desired accuracy in estimating the parameters of the GMM representing $F_{z_n}(\hat{\pi}_{n-1})$, i.e., the true filtered density given observation $z_n$ and the PGM posterior pdf at the previous time $\hat{\pi}_{n-1}$. Let $\delta'(n)$ be the desired confidence of the estimate. IF $\epsilon'(n)$ and $\delta'(n)$ are chosen such that:
\begin{align}
C(n)K(n)\epsilon'(n) d = \epsilon, \label{S4} \\
\delta'(n) = \frac{\delta}{N}, \label{S5}
\end{align}
and the corresponding number of samples $N_{\epsilon'(n), \delta'(n)} (n)$ be chosen according to Eqs. \ref{S1}-\ref{S3}, then it follows that $||Prob||\hat{F}_{z_n}(\hat{\pi}_{n-1}) - F_{z_n}(\hat{\pi}_{n-1})||> \epsilon) \leq \frac{\delta}{N}$.
\end{corollary}
\begin{proof}
Recall that $\hat{\pi}_n = F_{z_n}(\hat{\pi}_{n-1})$, and $\hat{\hat{\pi}}_n = \hat{F}_{z_n} (\hat{\pi}_{n-1})$. Then, from Lemma \ref{L6} we have that $||\hat{\pi}_n - \hat{\hat{\pi}}_n || \leq C(n)K(n) d\epsilon'(n)$ if $|\hat{\theta}_i(n) - \hat{\hat{\theta}}_i(n)| < \epsilon'(n)$ for all $i$, where $\hat{\theta}_i(n)$ represents the true parameters underlying the GMM representation of $\hat{\pi}_n$ and $\hat{\hat{\theta}}_i(n)$ represents their PGM approximation. Hence:
\begin{align}
Prob(||\hat{\pi}_n -\hat{\hat{\pi}}_n||> C(n)K(n)d\epsilon'(n)) > 1-\delta'(n),
\end{align}
which owing to the definition of $\epsilon'(n)$ and $\delta'(n)$ leads us to the desired result.
\end{proof}
Hence, using Corollary \ref{C7} and Lemma \ref{L2}, it follows that if the number of samples used to approximate the parameters of the predicted GMM pdf at time $n$ is greater than the $N_{\epsilon'(n), \delta'(n)}$, then it follows that $Prob(||\hat{\hat{\pi}}_n -\hat{\pi}_n|| > \frac{(1+\nu)C\epsilon}{1-\rho}) \leq \delta, $ for all $n$ for any arbitrarily small $\epsilon, \delta, \nu > 0$. However, in order for Assumption \ref{S_acc} to be valid, the estimates $\hat{\hat{\theta}}_n$ have to be almost surely bounded. To show this, due to the Strong Law of Large Numbers, it is also true that $\hat{\hat{\theta}}^N_n \rightarrow \hat{\theta}_n$ as $N \rightarrow \infty$, where  $\hat{\hat{\theta}}^N_n$ represent the estimate of the parameters after $N$ samples. Given the sample size is large enough, the estimate $\hat{\hat{\theta}}^N_n$ is arbitrarily close to the true parameters $\hat{\theta}_n$ almost surely, and thus, since the true parameters are bounded, so are the estimates.
This may be summarized in the following result.
\begin{proposition} \label{P1}
Let Assumptions \ref{UM_update} and \ref{EF} hold. Given a perfect clustering algorithem $\mathcal{C}$, and any $\epsilon, \delta, \nu >0$, at every time step $n$, choose the required accuracy of the approximation $\epsilon'(n)$ from Eq. \ref{S4}, the required confidence $\delta'(n)$ from Eq. \ref{S5}, and the corresponding minimum number of samples $N_{\epsilon'(n), \delta'(n)}$ from Eqs. \ref{S1}-\ref{S3}, then:
\begin{align}
Prob(||\hat{\hat{\pi}}_n -\hat{\pi}_n|| > \frac{(1+\nu)C\epsilon}{1-\rho}) \leq \delta . \nonumber\\
\end{align}
\end{proposition}
Several remarks are in order here.
\begin{remark}
The above result establishes the weak convergence of the approximate PGM filter density to the true filtered density uniformly over all time under the assumptions of exponential forgetting of the initial conditions and the adequacy of the Kalman update to approximate the true Bayesian update in the filtering equations. In the absence of the exponential forgetting condition, the convergence result can be obtained only for a finite number of time steps, the development being almost identical. In the absence of Assumption \ref{UM_update}, the adequacy of the Kalman update, the result  still holds except that there is a new error incurred in sampling the posterior pdf, which will be covered in the companion paper.
\end{remark}
\begin{remark}
The analysis above shows that the number of samples required at any time step to ensure the accuracy of the filter depends on the current predicted and posterior pdfs, and thus, in general, have to be time varying. This is a fact that is typically ignored in other mixture filters such as the PF and the GMF.
\end{remark}
\begin{remark}
\textbf{The Curse of Dimensionality:} It can be seen from the analysis above, in particular Eqs. \ref{S1}-\ref{S3}, that the number of samples required to estimate the parameters of the predicted and posterior pdfs accurately is completely independent of the dimension of the state space, and thus, is free from the "Curse of Dimensionality". However, we have to be more careful regarding the functional $L_1$ error in the PGM density: Eq. \ref{S4} shows that the accuracy parameter required at every time step is inversely proportional to the dimension of the since $\epsilon'(n) = \frac{\epsilon}{C(n)K(n)d}$, and thus, in order to attain the same accuracy in terms of the functional error of the filtered density, Eq.  \ref{S2} shows that the number of samples have to increase as $O(d^2)$ where $d$ is the dimension of the problem. Further, it should also be noted that the computation of the sample averages required by the PGM filter grows as $O(d^2)$.
\end{remark}
\begin{remark}
\textbf{Compressive Assumption:} The assumption of a finite Gaussian Mixture, is, in our opinion, a compressive argument. We restrict the set of parameters describing the predicted random variable to a finite number that have to be estimated using sample averages of the predicted random variable. The variance of these random samples is always bounded because of the finiteness assumption, and thus, the number of samples required to estimate the parameters is independent of the dimension of the problem. In general, if we were to find the moments of a random variable from its samples, we need to find all the moments via their sample averages. The variance of the samples of these higher order moments are, in general, not bounded, thereby requiring an infinite number of samples to estimate the pdf. 
\end{remark}
\subsection{Relationship to other Nonlinear Filters}
In this section, we compare and contrast the PGM filter with other nonlinear filters, in particular the PF, the EnKF and the GMF in detail.\\
The prediction stage of the PF is the same as the PGM except that the PF does not get a GMM from the set of predicted particles, and directly uses the Bayesian update on the individual particles, i.e, weights every particle with its likelihood $p(z_n/X_i)$. In the update step lies the computational trouble inherent to the PF, also known as the ``particle depletion problem": as the number of dimensions increase, it gets increasingly hard to sample particles with high likelihood, in fact the number of particles goes up exponentially with the number of dimensions thereby subjecting the PF to the curse of dimensionality. Consider the simple one dimensional example shown in figure \ref{fig:partdeplete}. In this case, a set of 400 particles are sampled from the prior pdf $\pi(x)=\ga(x,11,0.3)$. The measurement likelihood function is assumed to be $l_{z}(x)=\ga(x,15,0.1)$. Since the two pdfs are widely separated, nearly all the weight is allocated to a single particle as observed in the histogram of normalize weights in fig \ref{fig:partdeplete}.  Please see the references \cite{Beng,beng2,Beng3} for more rigorous insight. In contrast, the PGM uses the Kalman update for the GMM components and thereby does not suffer from the particle depletion problem. Moreover, as has been shown in the previous section, the number of samples required by the PGM is not dependent on the dimension of the problem. In essence, the Kalman update can be thought of as an automatic method to control/ move the predicted particles to the correct regions of the state space given the observation. In fact, this is the philosophy used in the EnKF\cite{Eveng,Geven} that perturbs each of the predicted particles using the measurement to obtain a perturbed ensemble of points that actually samples the posterior density. However, the EnKF always assumes a unimodal Gaussian for its predicted and posterior filter densities. At a more minor level, in PGM, we actually do the mean and covariance update of the components using the Kalman update equation rather than perturbing the ensemble of predicted particles\cite{Burgers}. Using the particle based Kalman Filtering method such as the EnKF, we see that at least the calculation of the mean and covariance is independent of the dimension of the state space, and this is precisely the reason why the EnKF is used regularly for the filtering of PDEs such as those arising in Meteorology and Geophysics \cite{Ander} where even the EKF or UKF are intractable.\\
\begin{figure}[t]
\centering
\includegraphics[width=0.45\textwidth,height=0.2\textheight]{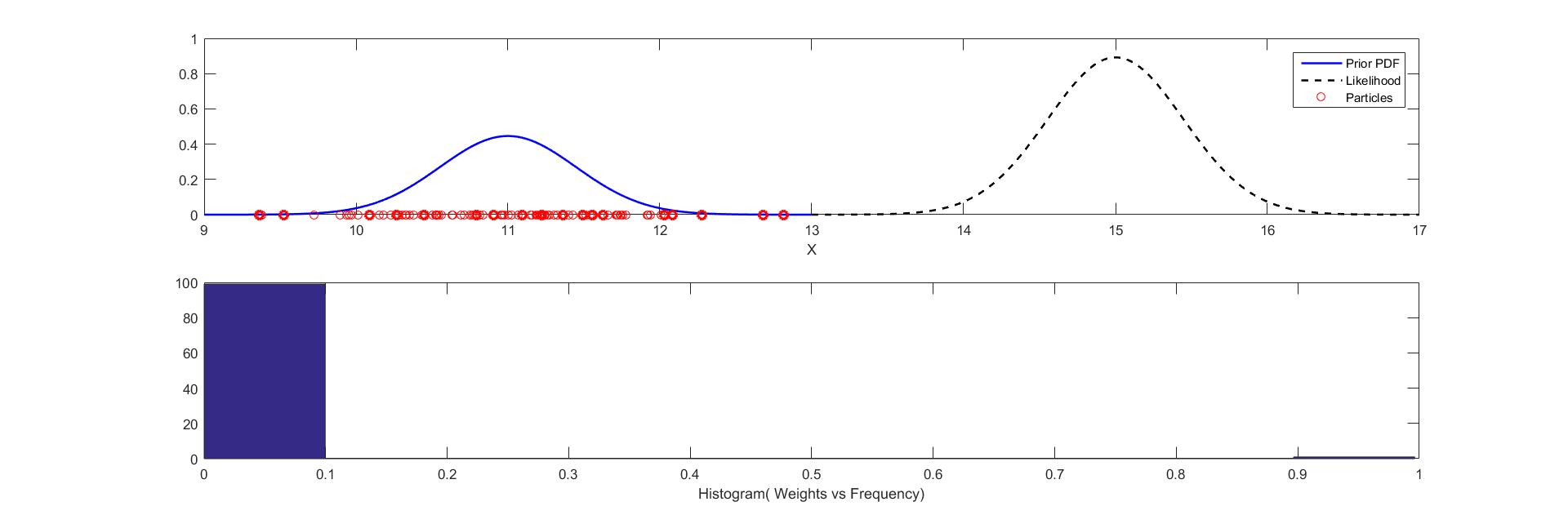}
\caption{Particle Depletion}
\label{fig:partdeplete}
\end{figure}

The earliest GMFs were  a bank of parallel EKFs but the number of modes in the GM was always fixed through both the prediction and the update steps \cite{Sorenson}. This, as has been noted in Section II, can be quite restrictive as it only considers multi-modality arising from initial conditions and never from the prediction and update steps. Other GMFs have more sophisticated methods for updating the weights of the GMM using the Fokker-Planck equation \cite{Terejanu} but keeps the number of modes fixed nonetheless. Recent GMF algorithms have also focused on time varying number of modes and used various heuristics to decide when to split a particular Gaussian into multiple components \cite{Mars}. Most of these methods use typical Kalman filtering propagation methods such as in the EKF/ UKF to propagate as well as split the Gaussian components. In contrast, we use a particle ensemble of the predicted random variable along with a clustering algorithm to conveniently find the number as well as the mean and covariances of the component clusters. In particular, we feel that the PGM harnesses the strength of the PF, the particle prediction step, along with the strength of the Kalman update in GMFs, using clustering algorithms, to develop a  technique that is free from the weaknesses of either technique.
\section{PGM Filter Implementation}
\label{sec:implem}
In this section, we discuss the  computations involved in the actual implementation of the PGM filter in detail. We separate these computations into the following four broad categories.
\begin{enumerate}
\item \textbf{Sampling}:
 Let the posterior pdf at time $n$ be given by 
\begin{equation}
\pi_n(x) = \sum_{i=1}^{M(n)} \omega_i(n) \ga_i(x; \mu_i(n), P_i(n))
\end{equation}
Draw an ensemble $S_{n}$ of $N_{p}$ states $\lbrace x_{n}^1,\cdots,x_{n}^{i},\cdots,x_{n}^{N_{p}} \rbrace$  from the GMM $\pi_{n}(X)$. The set $S_{n}$ represents the uncertainty involved in the state estimate at $n$. Draw $N_{p}$ independent samples of the process noise term $w(n)$ from its density $P_{W}(w)$ to get $Sw_{n}=\lbrace w_{n}^{1},\cdots,w_{n}^{i},\cdots,w_{n}^{N_{p}} \rbrace$. Let
\begin{equation}
x_{n+1}^{i-}=f(x_{n}^{i})+w_{n}^{i}.
\end{equation}
Then the set $S_{n+1}^{-1}=\lbrace x_{n+1}^{1-},\cdots,x_{n+1}^{i-},\cdots,x_{n+1}^{N-} \rbrace$   obtained is distributed according to $\pi _{n+1}^{-}(x)$ and represents the uncertainty in the propagated random variable in discrete form.
\item \textbf{Clustering}:
 Before proceeding to the measurement update step, a new GMM is fit to the set $S_{n+1}^{-}$ using a clustering algorithm of choice. However, if a measurement update is not to be performed, this ensemble $S_{n+1}^{-}$ is set to be $S_{n+1}$ which can be propagated further by drawing a new set of noise terms and substituting in the dynamic model as described above. The objective of clustering is to be able to compute a parametric mixture model describing the distribution of data in $S_{n+1}^{-}$ including the number of mixture components $M^{-}(n+1)$. In the present work, we have used the simple k-means clustering to compute GMM parameters. However, the k-means approach requires the total number of clusters to be specified externally. To work around this limitation, we have implemented a naive strategy which only requires the upper bound $M_{max}^{-}(n+1)$ as the external input instead of $M^{-}(n+1)$. We define the likelihood agreement measure ($L_{mes}$)\cite{Mars} as the measure of fitness of the parametric model $\theta_{a}$ in describing the dataset $S$. Let $\theta_{a,M}$be the M-component mixture model arrived at from $k-means$ clustering. Then the likelihood agreement measure may be computed as
\begin{equation}
L_{mes}(\theta_{a,M})=    \sum_{i=1}^{N_{p}}\pi_{\theta^{a}}(x^{i-}_{n+1})
\end{equation}
where $\pi_{\theta^{a}}(x)$ is the mixture pdf derived from the parametric model $\theta_{a,M}$. Let $\theta_{a^{*},M^{*}}$ be the optimal parametric model with $M^{-}_{n+1}=M^{*}$ components that maximized the likelihood agreement measure given the upper bound $M_{max}^{-}(n+1)$. Then, the naive strategy employed in implementing K-means clustering is presented in the following algorithm.
\begin{algorithm}
\caption{Clustering:Naive Strategy }
\begin{algorithmic}[1]
\REQUIRE $S_{n+1}^{-1}=\lbrace x_{n+1}^{1-},\cdots,x_{n+1}^{i-},\cdots,x_{n+1}^{N_{p}-} \rbrace ,M_{max}^{-}(n+1) $
\ENSURE $\theta_{a^*,M*},\hspace{5pt} M*\leq M_{max}^{-}(n+1)$
\STATE{$M\gets M_{max}^{-}(n+1)$}
\STATE{$\theta_{a^*,M*}\gets \theta_{a,M^{-}_{max}(n+1)}$}
\STATE{$L_{mes}^{*}\gets L_{mes}(\theta_{a,M^{-}_{max}})$}
\WHILE{$M>1$}
\STATE{$M \gets M-1$}
\STATE{Compute $\theta_{a,M}$ using k-means}
\IF{$L_{mes}(\theta_{a,M})\geq L_{mes}^{*}$}
\STATE{$\theta_{a^*,M*}\gets \theta_{a,M}$}
\STATE{$L_{mes}^{*}\gets L_{mes}(\theta_{a,M})$}
\ENDIF
\ENDWHILE
\end{algorithmic}
\end{algorithm}
The most common implementation of the k- means clustering approach is the Lloyd's algorithm\cite{lloyd}. The time complexity of the Lloyd's algorithm is known to be $O(N_{p}Mdi)$ where $N_{p}$ is the number of particles to be clustered ,$d$ the dimensionality of the state space, $M$ the number of clusters and $i$ is the number of iterations\cite{Pell}. In practice, the algorithm stops after several iterations finding a local optimum. Implementing the naive clustering strategy as described here will result in a quadratic time complexity in $M_{max}^{-}(n+1)$.\\
\item \textbf{Measurement Update}:
Let $\pi^{-}_{n+1}$ be the mixture pdf derived from the parametric model $\theta_{a^*,M*}$.
\begin{equation}
    \pi^{-}_{n+1}=\sum \limits_{i=1}^{M^{-}_{n+1}} w_{i}\ga_{i}(X,\mu_{i}^{-}(n+1),P_{i}^{-}(n)).
\end{equation}
Let $z_{n+1}$ be the measurement vector recorded at time $n+1$. The PGMF algorithm approximates the posterior pdf as a GMM with N components. The mixture parameters for the posterior pdf are computed in two stages. In the first stage, the component means and covariances are computed through Kalman measurement update of the individual mixture modes as described in Eqs.\ref{KU1},\ref{KU2}. However, in the present work we have considered two different approaches to computing the covariance terms($P_{i,ZX}(n+1),P_{i,ZZ}(n+1)$) and the expectations ($E_{i}(h(X))$) required for performing the Kalman update. 

\subsection{Update 1(PGM filter 1)}
In this approach, we compute the statistics of the posterior random variable with the unscented transform using a set of of $2d+1$ sigma points that are distributed symmetrically . 
%

The covariance terms and the expectations  required for computing the Kalman gain and posterior statistics are then computed as the weighted sample averages from the sigma points. Hence in PGM filter 1, the posterior means and covariances are computed by  performing a UKF measurement update individually on each GMM component.


\subsection{Update 2 (PGM filter2)}
In this approach, the covariances and cross covariances required for computing the gain matrix are evaluated directly from the particles.
Let $S_{j,n+1}^{-1}=\lbrace x_{j,n+1}^{1-},\cdots,x_{j,n+1}^{i-},\cdots,x_{j,n+1}^{Nj-} \rbrace$ denote the set of particles that form the $j-th$ cluster. Then the mean and covariance terms required from the cluster $j$ are computed as
\begin{align}
&z_l=g(x_{j,n+1}^{l}) \quad l=1,\hdots,Nj\\\nonumber
&\hat{z}_{j}=\frac{1}{Nj}\sum_{l=1}^{Nj} z_l\\\nonumber
&P_{j,ZZ}=\frac{1}{N-1}\sum_{l=1}^{Nj} \left(z_l-{\hat{z}_{j}}\right)\left(z_l-{\hat{z}}_{j}\right)^{'} + R\\\nonumber
&P_{j,ZX}=\frac{1}{N-1}\sum_{l=1}^{Nj}\left(z_l-{\hat{z}_{j}}\right) \left(x_{j,n+1}^{l}-{\mu_{j}^{-}(n+1)}\right)^{'}\\\nonumber
\label{eqn:gathe}
\end{align}

The posterior mean and covariance $\mu_{i}(n+1),P_{i}(n+1)$ for the remaining components are evaluated in a similar manner. Once all the component means and covariances are updated, the mixture weights $w_{i}$ can be computed as described below.
\begin{enumerate}
\item Construct the $ith$ mode measurement pdf $P_{Z,i}(z)$ given by
\begin{equation}
P_{Z,i}(z)=\ga_{i}(z,\hat{z}_{i},P_{i,ZZ})
\end{equation}
\item Evaluate the probability $l_{i}(n+1)$ given by
\begin{equation}
l_{i}(n+1)=P_{Z,i}(z_{n+1}),
\end{equation}
where $z_{k+1}$ is the measurement vector recorded at $t=n+1$.
\item The weight of the $ith$ posterior Gaussian mixture component is then given by
equation \ref{MU.D}\end{enumerate}
The weights, means and covariances as computed here characterizes the GMM describing the posterior pdf.\\
\item \textbf{Merging}:
 Depending on the clustering scheme, dynamics and measurement models, one may observe several closely distributed mixture modes in the posterior pdf. The components that are located sufficiently close may be merged to obtain a mixture model with well separated modes. A similar situation may arise when the clustering scheme assigns a complex model to describe the data due to overfitting. To identify the rights modes to be merged, we define the following normalized error metric \cite{Beck} as a measure of similarity between modes $i$ and $j$.
\begin{equation}
    D(i,j)=\frac{\int (\ga_{i}(x,\mu_{i},P_{i})-\ga_{j}(x,\mu_{j},P_{j}))^{2}dx}{\int \ga_{i}(x,\mu_{i},\Sigma_{i})^{2} dx +\int \ga_{j}(x,\mu_{j},\Sigma_{j})^{2}dx}
\end{equation}
Clearly, $D(i,j)=0$ when the components $i,j$ are identical and $D(i,j)=1$ when they are completely dissimilar. By evaluating the Gaussian integrals involved, the expression for $D(i,j)$ can be reduced to
\begin{equation}
    D(i,j)=\frac{|4\pi P_{i}|^{-1/2}+|4\pi P_{j}|^{-1/2}-2\mathcal{N}(\mu_{i},\mu_{j},P_{i}+P_{j})}{|4\pi P_{i}|^{-1/2}+|4\pi P_{j}|^{-1/2}},
\end{equation}
where $|.|$ represents the determinant of the enclosed square matrix.
Mixture modes that are closely spaced, can be merged whenever the value of normalized error metric falls below a predetermined tolerance ($tol$). In the present study, we have chosen this tolerance to be $tol=0.01$. Let $i_{1},\cdots,i_{k}$ be the indices of the mixture modes that are to be merged. Then the mixture parameters of the new Gaussian component obtained after merging is given by
\begin{flalign}
\omega_{i}=&\sum_{l=1}^{k} \omega_{i_{l}}\\
\mu_{i}=&\frac{\sum_{l=1}^{k} \omega_{i_{l}}\mu_{i_{l}}}{\omega_{i}}\\
P_{i}=&\frac{\sum_{l=1}^{k} \omega_{i_{l}}(P_{i_{l}}+(\mu_{i_{l}}-\mu_{i})(\mu_{i_{l}}-\mu_{i})^{T})}{\omega_{i}}
\end{flalign}
Recursive implementation of the prediction, clustering, update and merging algorithms as described here constitutes the PGM filter.
\end{enumerate} 
\section{Numerical Examples}
\label{sec:numex}
In this section, the particle Gaussian mixture filter is applied to three test case problems to evaluate the filtering performance. Other nonlinear filters such as the UKF and PF are also simulated for comparison. For the PF, a sequential importance resampling (SIR) design is considered. The estimation results are assessed for accuracy, consistency and informativeness. An exact description of the metrics used to compare the filter performance in each of the aforementioned categories is provided below.
\begin{enumerate}[label=(\alph*)]
\item Accuracy:  A Monte Carlo averaged root mean squared error ($E_{rms}(t)$) is considered for evaluating the accuracy of the estimates. The value of $E_{rms}(t)$ is computed as
\begin{equation}
E_{rms}(t)=\sqrt{\dfrac{1}{N_{Mo}} \sum_{j=1}^{N_{Mo}}\| \hat{X}^{j}(t)-\hat{\mu}^{j}(t)\|_{2}^{2}}.
\end{equation}

Here, $\hat{X}^{j}(t)$ and $\hat{\mu}^{j}(t)$ represent the actual and estimated states at the time instant $t$ during the $j$th Monte Carlo run. The time averaged error($\overline{E_{rms}}$) can be computed from $E_{rms}(t)$ as 
\begin{equation}
\overline{E_{rms}}=\dfrac{1}{T}\sum_{t=1}^{T}E_{rms}(t).\\
\end{equation}
\\

\item Consistency:
The consistency of the filtered pdf is examined using the normalized estimation error squared (NEES) test\cite{Bail}. For a unimodal state pdf, the NEES test is evaluated using the $\chi^{2}$ test statistic ($\beta_{j,t}$) given by
\begin{equation}
\beta_{j,t} = (\hat{X}^{j}(t) -\hat{\mu}^{j}(t))^{T} (\hat{P}^{j}(t))^{-1}(\hat{X}^{j}(t) -\hat{\mu}^{j}(t)).
\end{equation}

The term $\hat{P}^{j}(t)$ in the above expression represents the covariance of the unimodal filtered pdf at time $t$ during $j$th Monte Carlo run. The Monte Carlo averaged NEES test ($\beta_{t}$) is computed from this expression as 
\begin{equation}
    \beta_{t}=\frac{1}{N_{Mo}}\sum_{j=1}^{N_{Mo}}\beta_{j,t}.
\end{equation}
It can be shown that when the state vector $x \in \Re^d$ is normally distributed, the product $N_{Mo}\beta_{t}$ has a $\chi^{2}$ distribution with $dN_{Mo}$ degrees of freedom. As a result, the consistency of the filtered pdf can be tested by determining whether $\beta_{t}$ falls within probable bounds determined from the corresponding $\chi^{2}$ random variable. When a particle filter is used for estimation, the NEES test statistic is evaluated using the sample mean and sample covariance of the ensemble of states.\\
 The NEES test as presented here is not suitable for evaluating the consistency of a multimodal pdf. Let the filtered pdf at time $t=n$ be given by
 \begin{equation}
     \pi_{n}(x)=\sum_{i=1}^{M_{n}} \omega_{i}\ga_{i}(x;\mu_{i},P_{i}).
 \end{equation}
When the mixture modes are well separated, the total probability that the r.v represented by the GMM belongs to any one of the mixture modes is given by its mixture weight. In defining a GMM describing the state of the dynamical system, the filter hypothesizes the mixture weights, the  component means and their covariances. Hence, it is indispensable that, along with the means and covariances, the consistency test also checks for the agreement of the mixture weights with the observed data. A novel two step procedure for evaluating the consistency of a GMM pdf is developed as part of the present work. Let the random vector $V_{\mathbbm{1}}$ be defined as
\begin{equation}
 V_{\mathbbm{1}}=[\mathbbm{1}_{c}^{1}(x) \cdots \mathbbm{1}_{c}^{i}(x) \cdots \mathbbm{1}_{c}^{M_{n}}(x)]^{T} , 
\end{equation}
where $\mathbbm{1}_{c}^{i}(x)$ represents the indicator function which equals 1 when the state belongs to the $i^{th}$ mixture component and zero otherwise. Then assuming that the mixture modes are well separated, it can be deduced that
\begin{equation}
    E( V_{\mathbbm{1}})=[\omega_{1} \cdots \omega_{i} \cdots \omega_{M_{n}}]^{T}.
\end{equation}
It should be observed that the merging procedure described in the previous section helps to keep the modes well separated as it coalesces the components that are closely spaced. Define
\begin{flalign}
\epsilon_{v}=&V_{\mathbbm{1}}-[\omega_{1} \cdots \omega_{i} \cdots \omega_{M_{n}}]^{T},\\
\epsilon_{v}^{2}=&\epsilon_{v}^{T} \epsilon_{v}.
\end{flalign}
Here the value of $V_{\mathbbm{1}}$ is evaluated over a single instant, i.e, $\sum_{i=1}^{M_{n}} \mathbbm{1}_{c}^{i}(x)=1$.
Then, it can be shown that
\begin{align}
   & E(\epsilon_{v}^{2})=\sum_{i=1}^{M_{n}} \omega_{i}(1-\omega_{i}),\\
   & E(\epsilon_{v}^{2}- E(\epsilon_{v}^{2}))^2=\sum_{i=1}^{M_{n}} \omega_{i}(1-\omega_{i})((1-\omega_{i})^{3}+\omega_{i}^3)\\\nonumber&+\sum_{k}\sum_{\substack{j\\ j \neq k}}\omega_{j}\omega_{k}(\omega_{j}+\omega_{k}-3\omega_{j}\omega_{k})- (E(\epsilon_{v}^{2}))^{2}.
\end{align}

Let $V_{\mathbbm{1}}^{j}(t)$ be the above defined indicator vector computed at time $t$ during $j^{th}$ Monte Carlo run. At each instant, the state vector is assumed to belong to the component which has the highest likelihood given the truth. That is 
\begin{equation}
\mathbbm{1}_{c}^{i}(\hat{X}^{j}(t))=
\begin{cases}
\hfill 1, \hfill &  i=\arg\max \ga_{i}(\hat{X}^{j}(t),\mu_{i},P_{i})\\
\hfill 0, \hfill & \text{otherwise}\\
\end{cases}
\end{equation}
Let the sum $Sw_{t}$ be defined as
\begin{equation}
    Sw_{t}=\sum_{i=1}^{N_{Mo}} \dfrac{\epsilon_{v}^{j}(t)^{2}-E(\epsilon_{v}^{j}(t)^{2})}{\sqrt{N_{Mo} E(\epsilon_{v}^{j}(t)^{2}- E(\epsilon_{v}^{j}(t)^{2}))^2}}
\end{equation}
The expectations involved in this sum are computed using the mixture weights $\omega_{i}^{j}(t)$ at time $t$ during the run $j$. As $N_{Mo}$ becomes large, the sum $Sw_{t}$ converges in distribution to a standard Gaussian random variable, $Sw_{t} \xrightarrow{d} \mathcal{\ga}(x,0,1)$. Hence probability based bounds on the value of  $Sw_{t}$ can be computed from a standard normal distribution. Indeed, the first step in the two step procedure for testing consistency of GMM pdfs involves evaluating $Sw_{t}$ to determine whether it falls within the desired bounds. In the second step, a NEES test statistic is computed from the GMM  except that the mean and covariance of the most likely mode is used to evaluate the $\chi^{2}$ random variable,i.e.,
\begin{flalign}
\beta_{j,t} = &(\hat{X}^{j}(t) -\hat{\mu}_{i}^{j}(t))^{T} (\hat{P}_{i}^{j}(t))^{-1}(\hat{X}^{j}(t) -\hat{\mu}_{i}^{j}(t)),\\ \nonumber i=&\arg\max N(\hat{X}^{j}(t),\mu_{i}^{j}(t),P_{i}^{j}(t))
\end{flalign}
The $\chi^{2}$ test statistic obtained from the above expression may then be averaged over several Monte Carlo runs to perform the NEES test. This completes the two step consistency check for GMM pdfs.\\

\item Informativeness:
The informativeness of estimates furnished by an estimator is an important marker of its performance. A more informative  estimate almost always has a higher utility in comparison to a more uncertain estimate. In practice, several nonlinear filters inflate the covariance of the estimates to ensure consistency sacrificing informativeness in the process\cite{Lef2}. Two separate metrics are considered for evaluating the consistency of estimates in the present work. When the estimation errors are similar, a more informative state pdf can be expected to produce higher likelihoods given the true state.  The averaged likelihood of the truth over $N_{Mo}$ Monte Carlo runs may be computed as 
\begin{equation}
L(t)=\frac{1}{N_{Mo}}\sum_{j=1}^{N_{Mo}}\pi_{t}^{j}(\hat{X}^{j}(t)).
\end{equation}
Here $\pi_{t}^{j}$ represent the state pdf conditioned upon all measurements recorded until the time instant $t$ in the $j$th Monte Carlo run. The time averaged likelihood is computed from the above expression as
\begin{equation}
\hat L=\dfrac{1}{T}\sum_{t=1}^{T}L(t).
\end{equation}
When the state pdf is represented as an ensemble of states, the likelihood is computed using a unimodal Gaussian pdf characterized by the sample mean and covariance of the collection of states. For a unimodal Gaussian pdf, the total volume occupied by the $r$-sigma ellipse contains a fraction $f_{p}$ of the total probability. It can be shown that this fraction depends only upon the value of $r$ and the dimensionality $d$ of the state space, i.e., $f_{p}=f_{p}(r,d)$. Let $f_{p}'$ be the fraction of probability contained in the $r'$-sigma ellipse for a fixed $r'>0$. Then the volume occupied by these $r'$ sigma ellipses may be used as a measure of the uncertainty in the state estimate with a larger volume indicating less informative estimates. For a GMM with  well separated modes, the volume of the state space that contains the fraction $f_{p}'$ of total probability may be computed as the sum of volumes of $n'$ sigma ellipses of individual mixture components. In the present work, the volume of state space $V^{j}\sigma_{2}$ that contains the fraction $f_{p}=f_{p}(2,d)$ of total probability is considered for comparing informativeness. For a well separated GMM pdf, the value of $V\sigma_{2}$  can be computed as
\begin{equation}
    V^{j}\sigma_{2}(t)= \sum_{i=1}^{M_{t}}|2\Sigma_{i}|,
\end{equation}
where $M_{t}$ is the number of modes.
Here $|.|$ represents the determinant of the enclosed square matrix. The expression for the unimodal case can be derived by setting $M_{t}=1$ in the above equation. We compute the Monte Carlo averaged $2-\sigma$ volume as
\begin{equation}
V\sigma_{2}(t)=\sum_{j=1}^{N_{Mo}}V^{j}\sigma_{2}(t).
\end{equation}
We also compute the corresponding time averaged value $\hat{V}\sigma_{2}$ given by
\begin{equation}
\hat{V}\sigma_{2}=\dfrac{1}{T}\sum_{i=1}^{T}V\sigma_{2}(t).
\end{equation} 
\begin{figure}[t]
\centering
\includegraphics[width=0.45\textwidth,height=0.2\textheight]{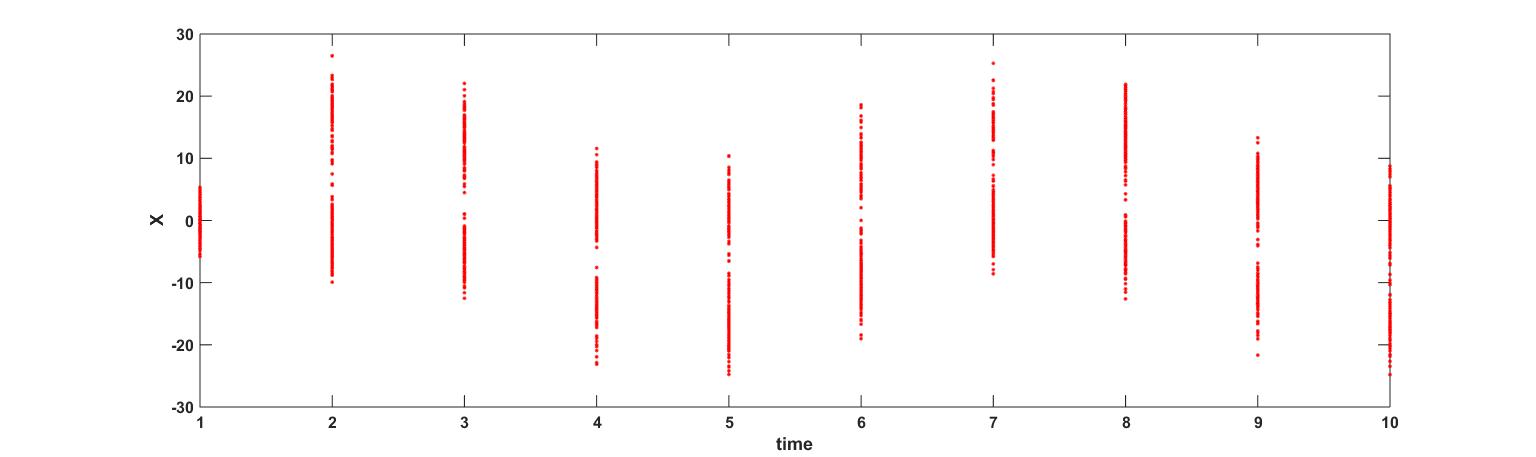}
\caption{Example 1:Propagation of states}
\label{fig:Ex1prop}
\end{figure}
\end{enumerate}
\subsection{Example 1}  Consider the discrete time nonlinear dynamic system given by
\begin{equation}
x_{k+1}=\frac{x_{k}}{2}+ \frac{25x_{k}}{(1+x_{k}^2)}+8\cos(1.2k)+\nu_{k}
\label{eqn:eq1}
\end{equation}
A measurement model aiding the estimation of the system is specified as follows.
\begin{equation}
z_{k}=\frac{x_{k}^2}{20}+n_{k}
\end{equation}

The process noise term $\nu_{k}$ and measurement noise term $n_{k}$ are assumed to be independent zero mean Gaussian random variables with covariances $Q$ and $R$ respectively.
Removing the cosine forcing term and setting $x_{k+1}=x_{k}$ in equation, we get the equilibrium points \\
\begin{align} 
            x_{k}&=\pm 7,0.
\end{align}
However, only the equilibrium  points at $\pm7$ are stable. Location of a set of 200 particles that are sampled and propagated from the pdf $\mathcal{N}(0,5)$ is provided in figure \ref{fig:Ex1prop} for illustrating the propagation of uncertainty in this system. It is worth noting that the measurement model does not disambiguate between $\pm x_{k}$ given a sufficiently uncertain prior pdf.
\overfullrule=2cm

Two variants of the PGM filter, an SIR filter and a UKF are simulated to estimate the test case 1 system for a duration of 52 time steps. The uncertain initial state of the system is assumed to distributed as
$P_0(x)=\mathcal{N}(0,2)$. The process and measurement noise covariances are set to be $Q=10, R=1$. Measurements are recorded at every other instant. The estimation process is repeated over 50 Monte Carlo runs and the time and Monte Carlo averaged performance metrics are computed. The SIR and the PGM filters are implemented with a set 50 particles in this one dimensional test case problem. The upper bound on the number of mixture components $M_{max}$ is set to be 2. The parameter values used in the implementation of the UKF may be found in Table \ref{tab:ukfex1}.  \\
\begin{table}
\centering
\caption{}
\label{tab:ukfex1}
   \begin{tabular}{ |c|c|c| }
     \hline
  \multicolumn{3}{|c|}{UKF parameters} \\
  \hline
  $\alpha$ & $\beta$ & $\lambda$ \\
\hline
  1.3 & 1.5 & 0.2 \\
  \hline
\end{tabular}
\end{table}

Figure \ref{fig:ex1tim} shows the sequence of actual and estimated states obtained from a single Monte Carlo run. The PGM filters and the PF are seen to offer superior tracking performance in comparison to UKF. The Monte Carlo averaged root mean squared error($E_{rms}(t)$) plotted in figure \ref{fig:ex1rmse} provides further support to this observation. As $x_{j,t}$ is a one dimensional random variable, the sum $N_{Mo}\beta_{t}$ is a $\chi^{2}$ random variable in $N_{Mo}$ dimensions. For 50 monte carlo runs, 99 per cent probability upper bound of the random variable $\frac{N_{Mo}\beta_{t}}{N_{Mo}}$ is found to be $Ub_{0.99}=1.5231$. For the multimodal filters, value of $\beta_{t}$ is computed from the covariance of the most likely mode as mentioned before. The Monte Carlo averaged NEES results along with the $Ub_{0.99}$ are plotted in figure \ref{fig:ex1nees}. It is observed that the UKF and PF frequently oversteps the upper bound which marks inconsistent estimates. Furthermore, $\beta_{t}$ computed using the PF estimates are found to frequently exhibit peaks several orders of magnitude larger than the $99\%$ upper bound, indicating covariance collapse. To study the informativeness of the  estimates, the averaged likelihood $L(t)$ is computed and plotted in figure \ref{fig:ex1likelihood}. The estimates provided by the PGM filters are seen to hold the highest likelihoods during most of the simulation time. A similar trend is observed when the volume $V_{2\sigma}$ is plotted against time in figure \ref{fig:ex1volume} where the PGM filters are seen to have the lowest $V_{2\sigma}$ during a large fraction of the simulated time. 
 \begin{figure}[t]
 \centering
 \includegraphics[width=0.45\textwidth,height=0.25\textheight]{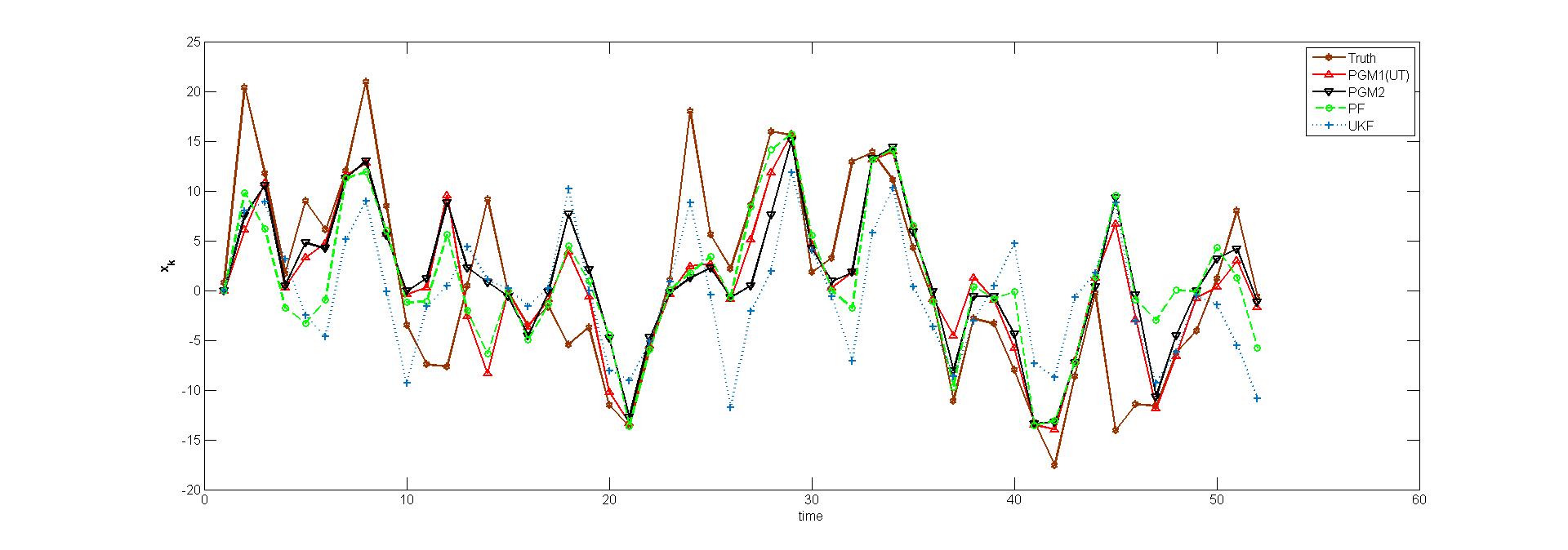}
 \caption{Time history of actual and estimated states (case 1)}
 \label{fig:ex1tim}
 \end{figure}

 \begin{table}[h]  
\centering
\caption{}
\label{tab:tab2}
   \begin{tabular}{ |c|c|c|c|c| }
     \hline
  \multicolumn{5}{|c|}{Example 1: Results} \\
  \hline
 & $\overline{E_{rms}}$ & $\beta_{t,c}(\%)$ & $\hat L$ & $\hat{V}\sigma_{2}$ \\ 
\hline       
  PGM 1(UT) & 6.3169 & 80.77 & 0.1153 & 63.4740\\
  \hline
  PGM 2  & 6.4223& 78.85 & 0.1167 & 61.8611\\
  \hline
  PF & 6.4580 & 42.31 & 0.1072 & 77.1697 \\
  \hline
  UKF  & 8.1980 & 28.85 & 0.0506 &101.034 \\
  \hline
\end{tabular}
\end{table}

To gain further insight, the time averaged values of  RMSE $\overline{E_{rms}}$, likelihood $\hat{L}$, and the $2 \sigma$ volume for each filter are listed in table \ref{tab:tab2}. Also included is the fraction ($\beta_{c\%}$) of the time instants during which the computed averaged NEES result stayed within the 99 percent consistency limits,i.e.,
\begin{equation}
    \beta_{c\%}=\frac{\sum_{t=1}^{T}\mathbbm{1}_{Ub0.99}(\beta_{t})}{T}
\end{equation}
where $\mathbbm{1}_{Ub0.99}(\beta_{t})$ is the indicator function which equals 1 when $\beta_{t}<Ub_{0.99}$ and zero otherwise. The consistency fractions for the PGM filters are seen to be almost double that of the PF. For the PGM filters, the component weights are also tested for consistency. For the PGM1 filter, the value of $Sw_{tZ}$ is found to stay within the 99 percent bounds during $80.38 \%$ of the simulated time. For PGM filter 2, this number was found to be $73\%$.  The time averaged RMSE for the PGM filters are seen to be the lowest and is closely followed by the PF. The PGM filters registered the highest averaged likelihoods and the lowest $2-\sigma$ volumes.
The results presented in table \ref{tab:tab2} clearly show that the PGM filter implementations offer the most accurate, consistent and informative estimates among the all the filters that were tested. 
\begin{figure}[h!]    
\begin{subfigure}[t]{0.45\textwidth}
\includegraphics[width=\linewidth,height=4cm]{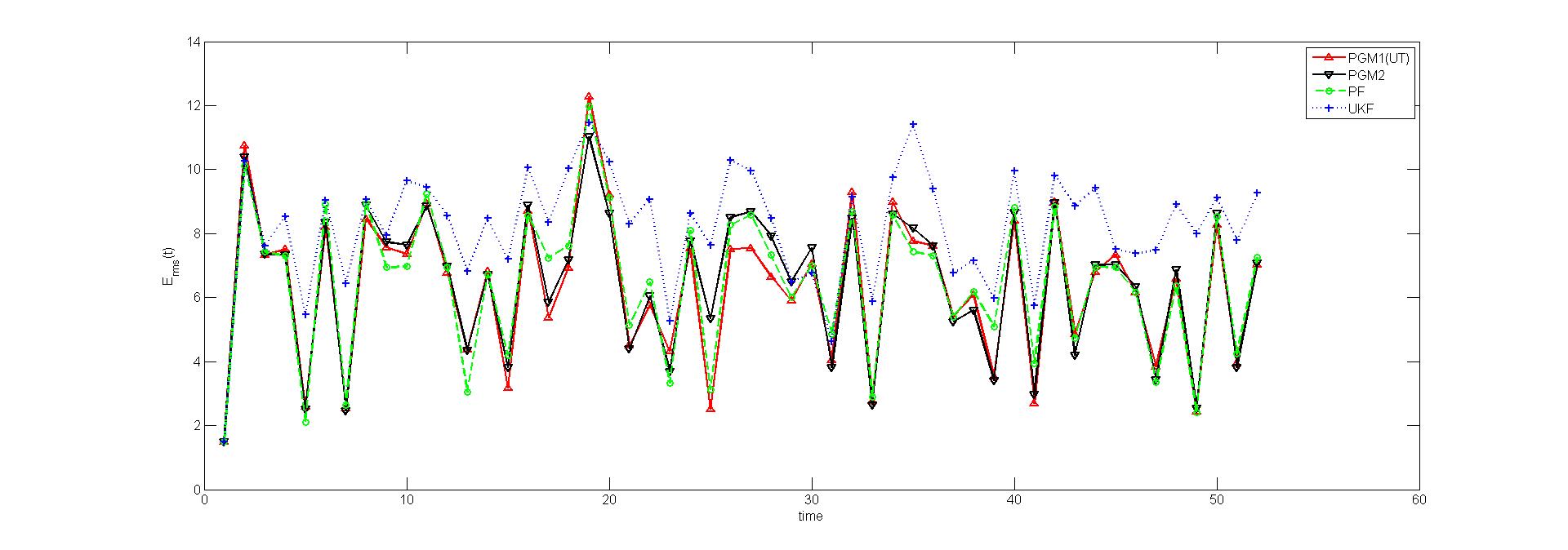}
\caption{$E_{rms}(t)$}
\label{fig:ex1rmse}
\end{subfigure}
\hspace{\fill}
\begin{subfigure}[t]{0.45\textwidth}
\includegraphics[width=\linewidth,height=4cm,]{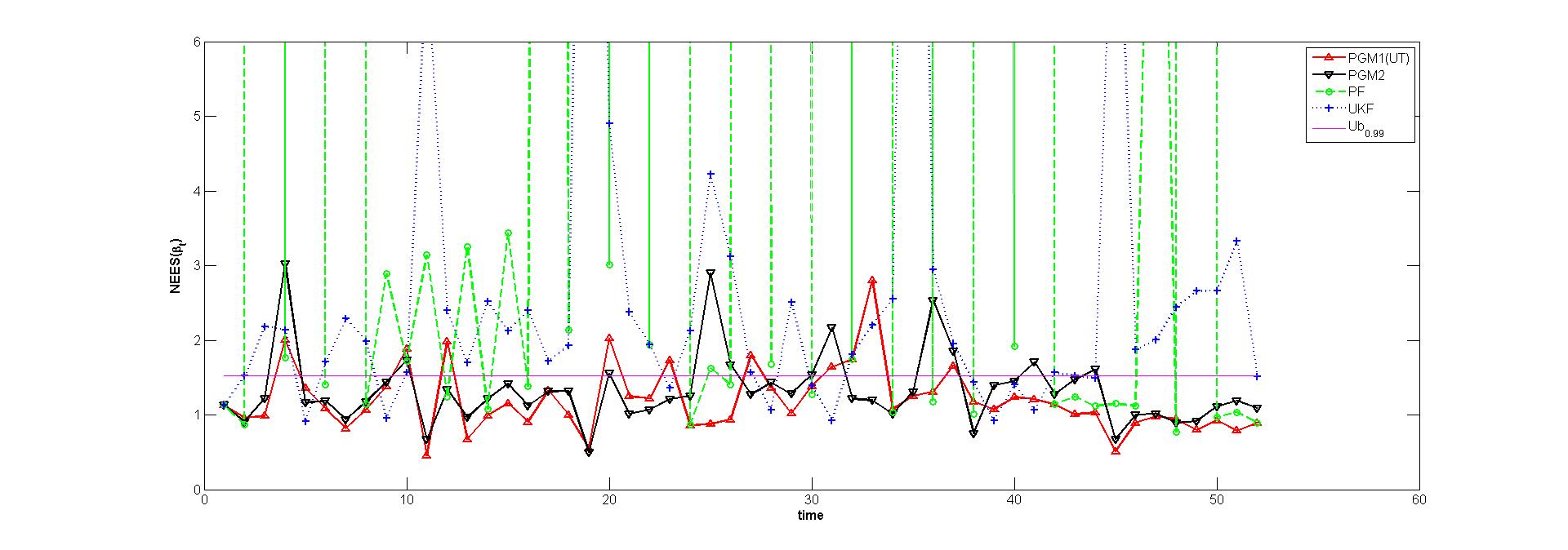}
\caption{$\beta_{t}$}
\label{fig:ex1nees}
\end{subfigure}
\begin{subfigure}[t]{0.45\textwidth}
\includegraphics[width=\linewidth,height=4cm,]{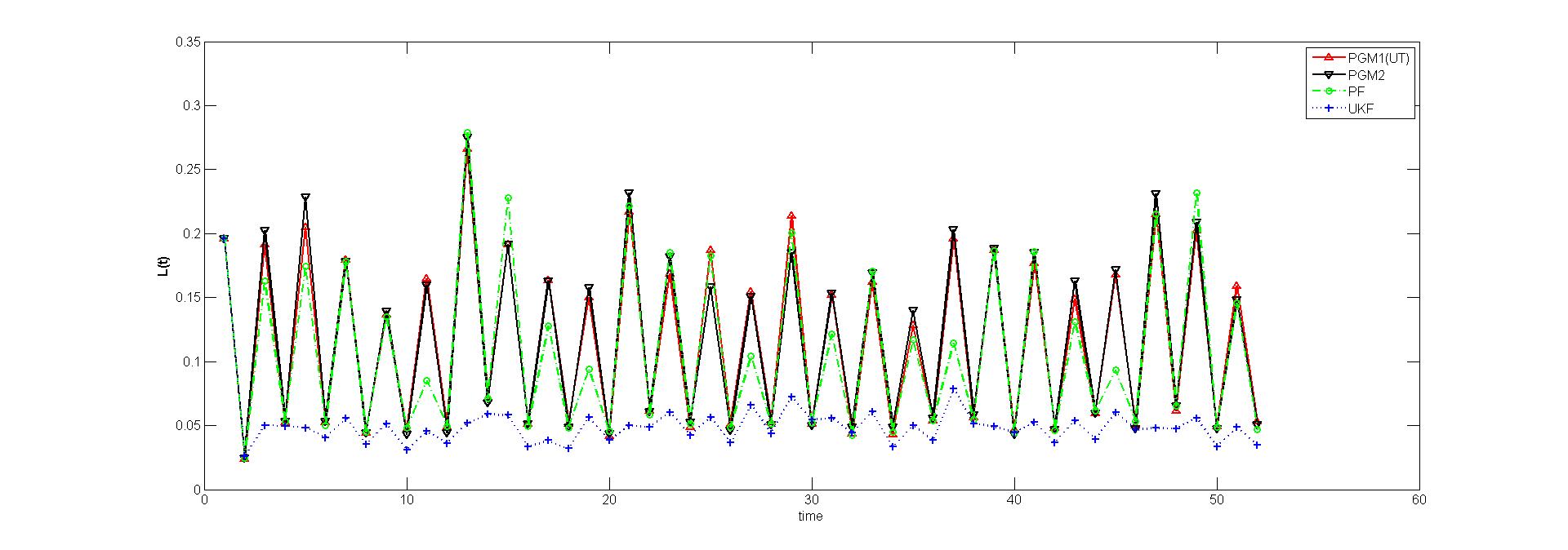}
\caption{$L(t)$}
\label{fig:ex1likelihood}
\end{subfigure}
\hspace{\fill}
\begin{subfigure}[t]{0.45\textwidth}
\includegraphics[width=\linewidth,height=4cm,]{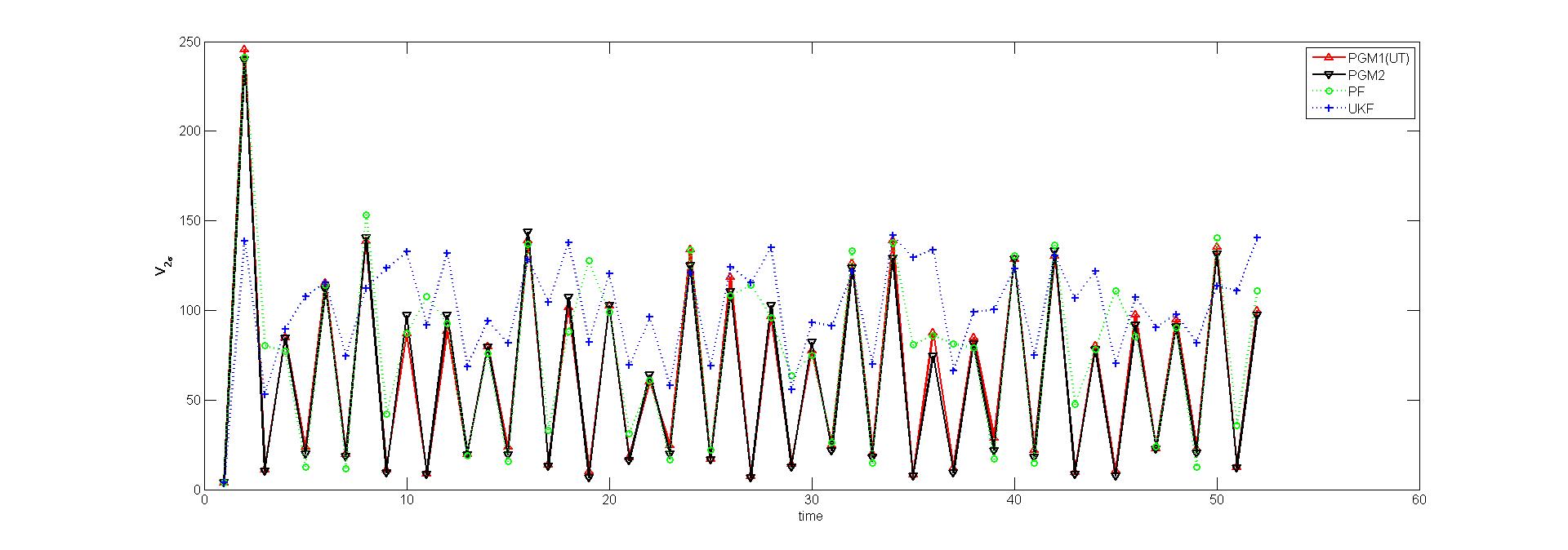}
\caption{$V_{2\sigma}$}
\label{fig:ex1volume}
\end{subfigure}
\caption{Example 1 Results}
\end{figure}

\subsection{Example 2}

In this example, the PGM filters are employed in the estimation of a 3 dimensional Lorenz 63 model for atmospheric convection. The noise perturbed dynamics of the Lorenz 63 system is described the the following set of equations,
\begin{flalign}
\dot{x_{1}}=&\alpha (-x_{1}+x_{2}),\\ \nonumber
\dot{x_{2}}=&\beta x_{1}-x_{2}-x_{1}x_{3},\\ \nonumber
\dot{x_{3}}=&-\gamma x_{3}+x_{1}x_{2}+\Gamma(t),\\\nonumber
\alpha=&10,\beta=28,\gamma=8/3 \nonumber.
\end{flalign}
A scalar nonlinear measurement model($z_{k}$) is considered which is given by
\begin{equation}
    z_{k}=\sqrt{x_{1}(t)^{2}+x_{2}(t)^{2}+x_{3}(t)^{2}}+\nu_{k}.
\end{equation}
The process and measurement noise covariances are both set be equal to 1.The initial state of the system is assumed to be uncertain and is characterized by the bimodal pdf
\begin{align}
    p_{0}(x)=&0.9\ga(x,[-0.2,-0.2,8]^{T},\sqrt{0.35}I_{3\times3})+\\\nonumber
    &0.1\ga(x,[0.2,0.2,8]^{T},\sqrt{0.35}I_{3\times3})
\end{align}
The state of the system is updated at a fixed time step $\Delta t=0.01 s$. The measurements are recorded at the interval of ten time steps. The Lorenz 63 system is known to exhibit chaotic solutions for the parameter values considered in the present simulation. Figure \ref{fig:ex2traj} shows the trajectories followed by a set of 500 particles that are sampled from the initial pdf and propagated through the dynamics of the system. The sampled trajectories are seen to split into two sets that occupy different regions of the state space.
\begin{figure}[t]
\centering
\includegraphics[width=0.45\textwidth,height=0.2\textheight]{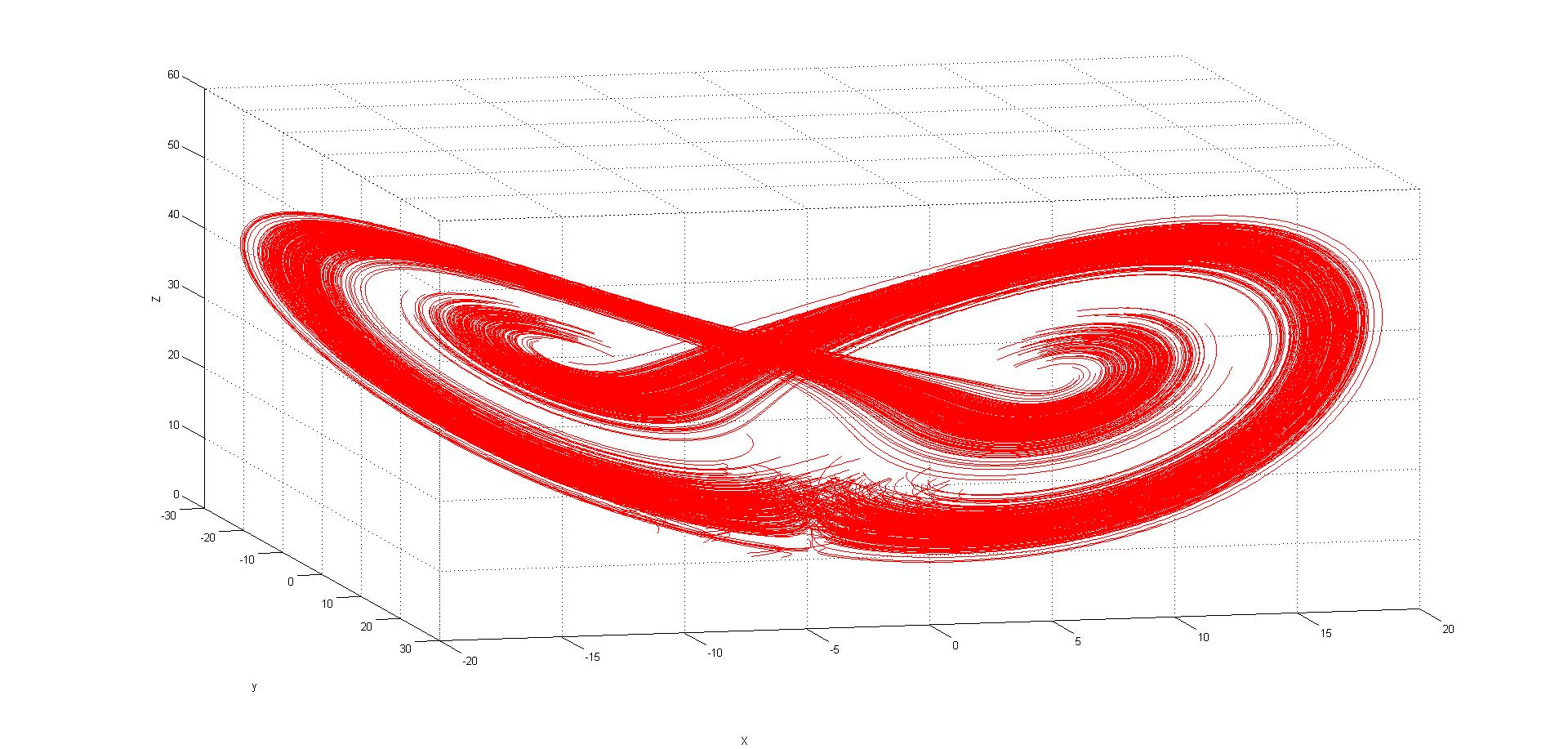}
\caption{Example 2:Propagation of states}
\label{fig:ex2traj}
\end{figure}
 \\

The two variants of the PGM filter, the PF and a conventional Gaussian mixture UKF are employed in the estimation of the Lorenz63 system. The PGM filters and the SIR filter are implemented with 300 particles. The maximum number of mixture components ($M_{max}$) is set be 2. The UKF is implemented using the parameters listed in table \ref{tab:ukfex1}. The dynamics and the estimation process are repeated over 50 Monte Carlo runs. The Monte Carlo averaged mean squared error of estimation $E_{rms}(t)$ for each filter is plotted in figure \ref{fig:ex2rmse}. The computed values of $E_{rms}(t)$ for PGM1 and PGM2 are found to be the smallest among the four filters that are tested. For the Lorenz63 system, when $N_{Mo}=50$, the $99\%$ probability upper bound on the random variable $\beta_{t}$ is found to be $Ub_{0.99}=3.8642$. For the PGM filters and the mixture UKF, the value of $\beta_{t}$ is computed using the covariance of the mixture component with the highest likelihood. The Monte Carlo averaged NEES results are plotted in figure \ref{fig:ex2nees}. The line $y=Ub_{0.99}$ has  also been included for reference. It is observed that the NEES test statistic $\beta_{t}$ for the PF and the mixture UKF overstep the $y=Ub_{0.99}$ line early in the simulation. It is also seen that, once this upper bound is crossed, the value of $\beta_{t}$ for these two filters remain outside the $99\%$ probability limits for the entire remaining duration of the test.
\begin{figure}[h!]    
\begin{subfigure}[t]{0.45\textwidth}
\includegraphics[width=\linewidth,height=4cm]{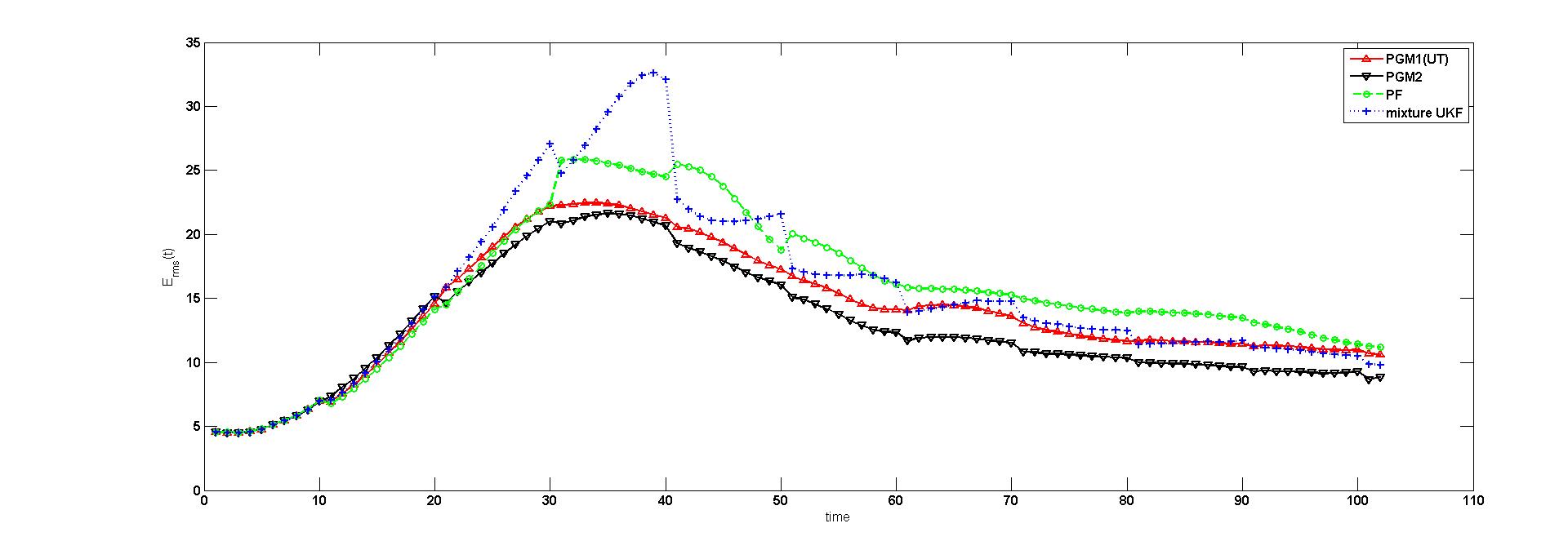}
\caption{$E_{rms}(t)$}
\label{fig:ex2rmse}
\end{subfigure}
\hspace{\fill}
\begin{subfigure}[t]{0.45\textwidth}
\includegraphics[width=\linewidth,height=4cm,]{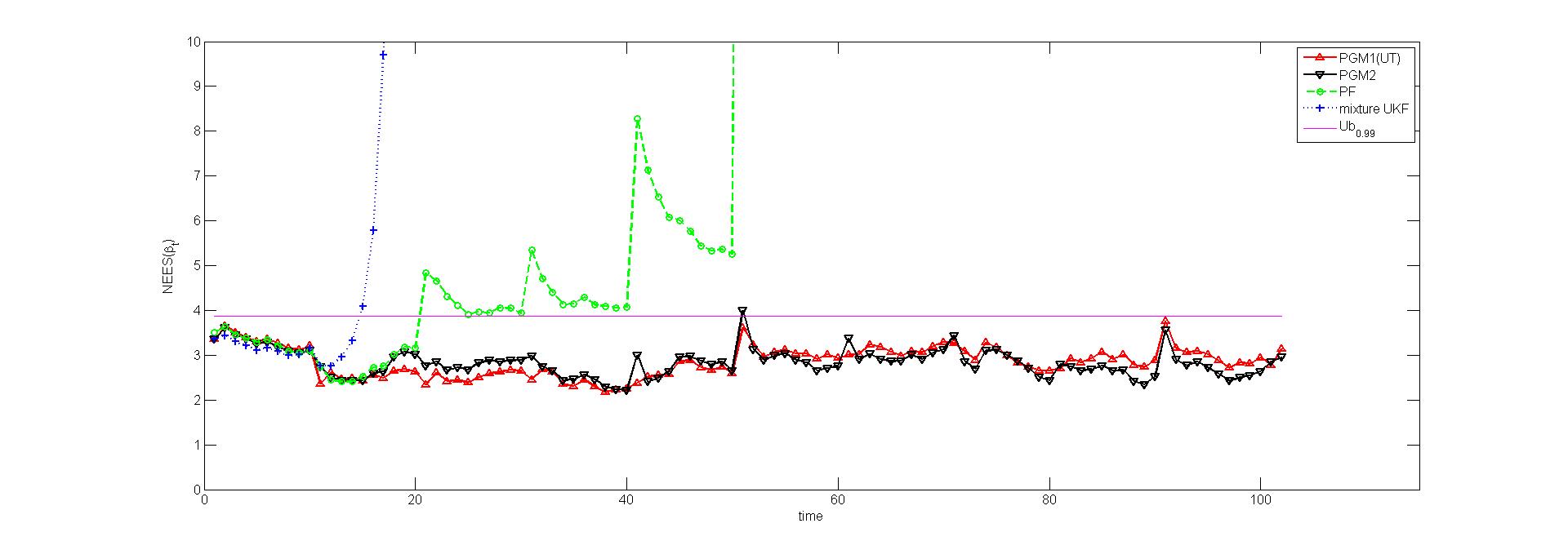}
\caption{NEES($\beta_{t}$)}
\label{fig:ex2nees}
\end{subfigure}

\begin{subfigure}[t]{0.45\textwidth}
\includegraphics[width=\linewidth,height=4cm,]{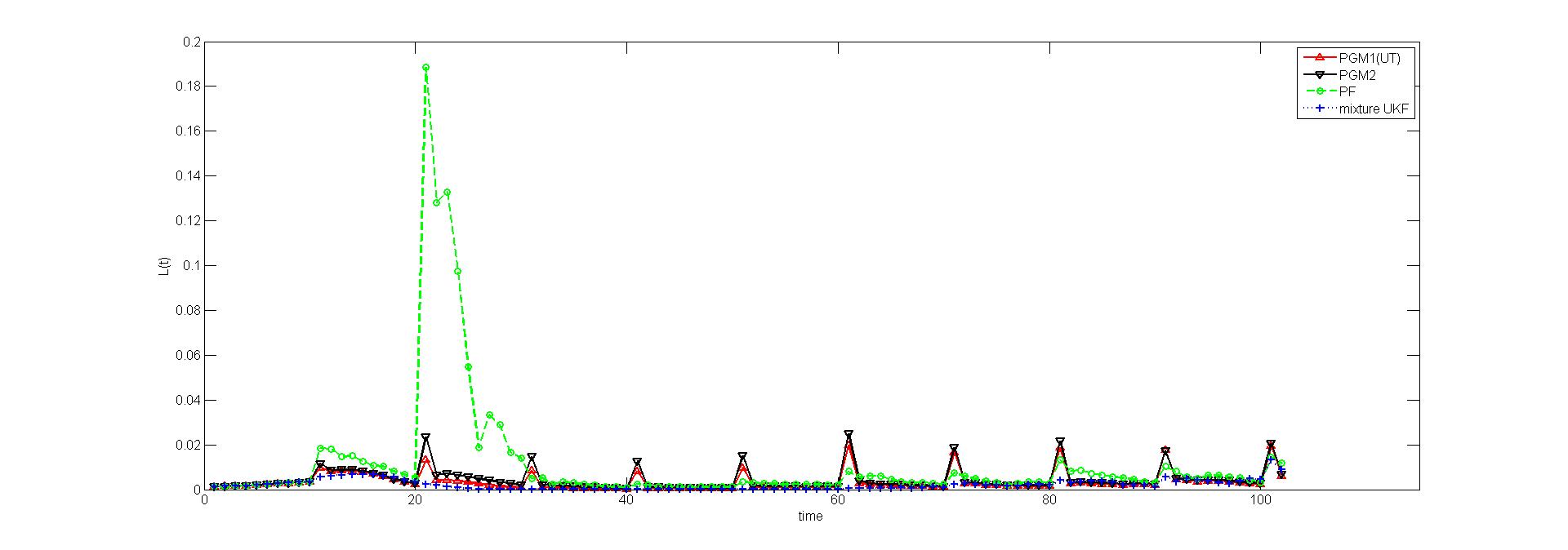}
\caption{$L(t)$}
\label{fig:ex2likelihood}
\end{subfigure}
\hspace{\fill}
\begin{subfigure}[t]{0.45\textwidth}
\includegraphics[width=\linewidth,height=4cm,]{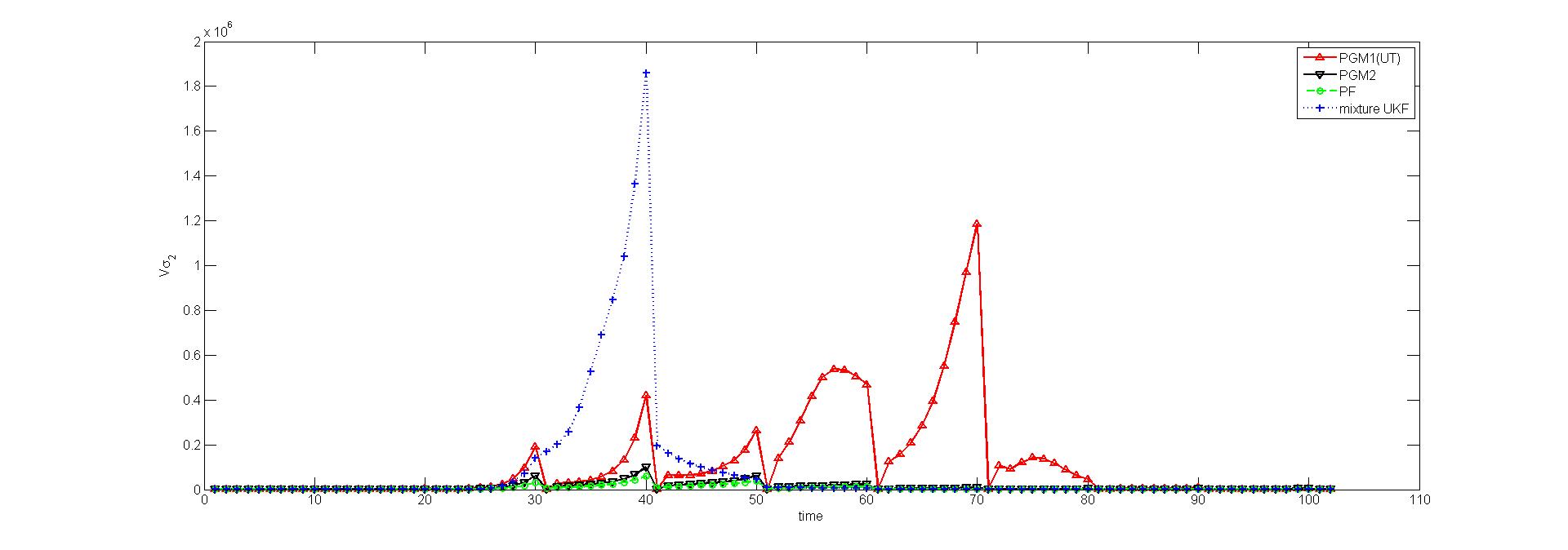}
\caption{$V_{2\sigma}$}
\label{fig:ex2volume}
\end{subfigure}
\caption{Example 2 Results} 
\end{figure}
 In contrast, he averaged NEES test statistic for the PGM filters are seen to stay within the $99\%$ during most of the simulation time. In order to compare the informativeness of the estimates, the averaged likelihoods($L(t)$) and $V\sigma_{2}$ volumes are computed and plotted in figures \ref{fig:ex2likelihood}, \ref{fig:ex2volume} respectively. Curiously, with the highest average likelihoods and the smallest $V\sigma_{2}$ volume, the PF is seen to outperform the PGM filters in furnishing the most informative estimates. However, as the NEES results of the PF
are seen to stay above $10^{2}$ during $40\%$ the simulation time, the higher likelihoods and the small $V\sigma_{2}$ of the PF should be understood as a consequence of its covariance collapse. The consistency fractions($\beta_{c\%}$) and the time averaged values of  RMSE $\overline{E_{rms}}$, likelihood $\hat{L}$, and $2 \sigma$ volume for each filter are listed in table \ref{tab:tab3}

\begin{table}[h]  
\centering
\caption{}
\label{tab:tab3}
   \begin{tabular}{ |c|c|c|c|c| }
     \hline
  \multicolumn{5}{|c|}{Example 2: Results} \\
  \hline
 & $\overline{E_{rms}}$ & $\beta_{t,c}(\%)$ & $\hat L$ & $\hat{V}\sigma_{2}(\times 10^{5})$ \\ 
\hline  
  PGM 1(UT) &   13.9812 & 100 & 0.0035  & 1.1508\\
  \hline
  PGM 2 & 12.7406& 99.02  & 0.0044 & 0.1051\\
  \hline
  PF &   15.5565 & 19.61 & 0.0114  &  0.0667  \\
  \hline
  GMUKF  & 15.3695& 13.73 &  0.0021 & 0.8564\\
  \hline
\end{tabular}
\end{table}
For the three multimodal filters PGM1, PGM2 and mixture UKF, the computed value of $Sw_{tz}$ is found to stay within the $99\%$ bounds during 81.82, 90.91 and 10 percent of the times considered.  The results clearly indicate that the PGM filters are more accurate and consistent than the PF and mixture UKF. The PGM filter estimates are also seen to be more informative than the mixture UKF estimates.

\subsection{Example 3}
In this test case, The PGM filters are employed in the estimation of a Lorenz96 system. The noise perturbed dynamics of the Lorenz96 system is given by
\begin{equation}
    \dot{x_{i}}=x_{i-1}(x_{i+1}-x_{i-2})-x_{i}+F+\Gamma(t),
\end{equation}
where $i=1,2,\cdots,40$. The term F represents a constant external forcing. In the present work, we set $F=8$. The covariance of the zero mean Gaussian white noise is assumed to be $Q=10^{-2}$. A linear measurement model is employed in the estimation of the Lorenz96 system and it is defined as,
\begin{flalign}
    z_{k}=&H X_{k}+\nu_{k},\hspace{8pt} H \in \mathbbm{R}_{20 \times 40}\\
    H_{i,j}=&\begin{cases} \nonumber
\hfill 1, \hfill &  j=2i-1\\
\hfill 0, \hfill & \text{otherwise}.\\
\end{cases}
\end{flalign}
Therefore, the measurement model records the components of the state vector that have odd indices. The measurement noise is assumed to be a zero mean Gaussian random vector with a covariance $R=10^{-2}I_{20 \times 20}$ where $I_{i,j}=\delta_{i,j}$. The initial state of the system is characterized by the pdf
\begin{flalign}
p_{0}(x)=&\ga(x,\mu_{0},P_{0}),
\end{flalign}
where $\mu_{0}=F\begin{bmatrix}1\cdots1\cdots1\end{bmatrix}^{T},\mu_{0} \in \mathbbm{R}_{40 \times 1}$ and $P_{0}=10^{-3}I_{40\times40}$. The time step for updating the state of the system is fixed at $\Delta t=0.05 s$. Measurements are recorded at the interval of 1s.  When a UKF is employed to estimate the state of the Lorenz96 system, the covariance of the estimate is frequently observed to become non positive definite. Hence, the UKF is not included in this comparison study. Instead, an EnKF is used to estimate the Lorenz96 system along with the SIR and PGM filters. All four filters are equipped with a set of 2000 particles. The value of $M_{max}$ is kept at 2. The filters are used to estimate the state of the system for a duration of 10s. The process is repeated over 50 Monte Carlo runs and the averaged performance metrics are computed. The accuracy of the estimates is compared using the Monte Carlo averaged RMSE ($E_{rms}(t)$)which is plotted in figure \ref{fig:ex3rmse}. From the $E_{rms}(t)$ plots, it can be observed that the tracking performance of the PGM filters and the EnKF are comparable. In comparison to the PGM filters, the PF is found to offer inferior tracking performance.  In order to compare the consistency of the filters, the value of the Monte Carlo averaged NEES test statistic $\beta_{t}$ is computed and plotted in figure\ref{fig:ex3nees}.
\begin{figure}[h!]    
\begin{subfigure}[t]{0.45\textwidth}
\includegraphics[width=\linewidth,height=4cm]{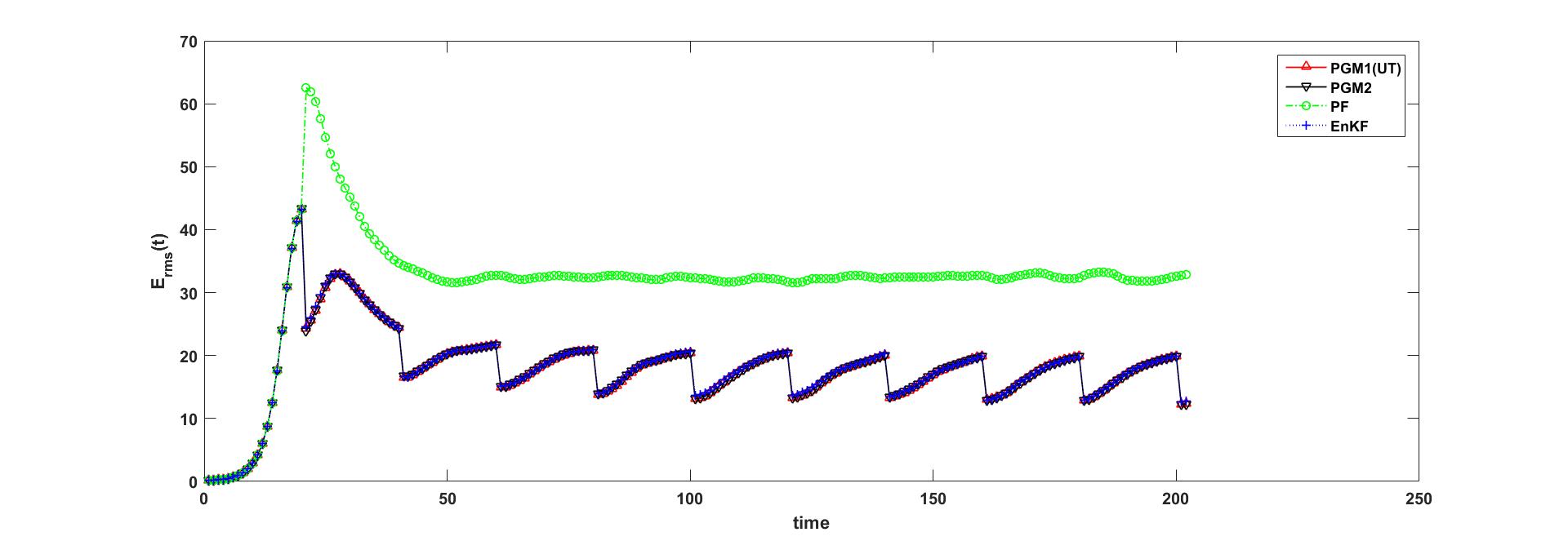}
\caption{$E_{rms}(t)$}
\label{fig:ex3rmse}
\end{subfigure}
\hspace{\fill}
\begin{subfigure}[t]{0.45\textwidth}
\includegraphics[width=\linewidth,height=4cm,]{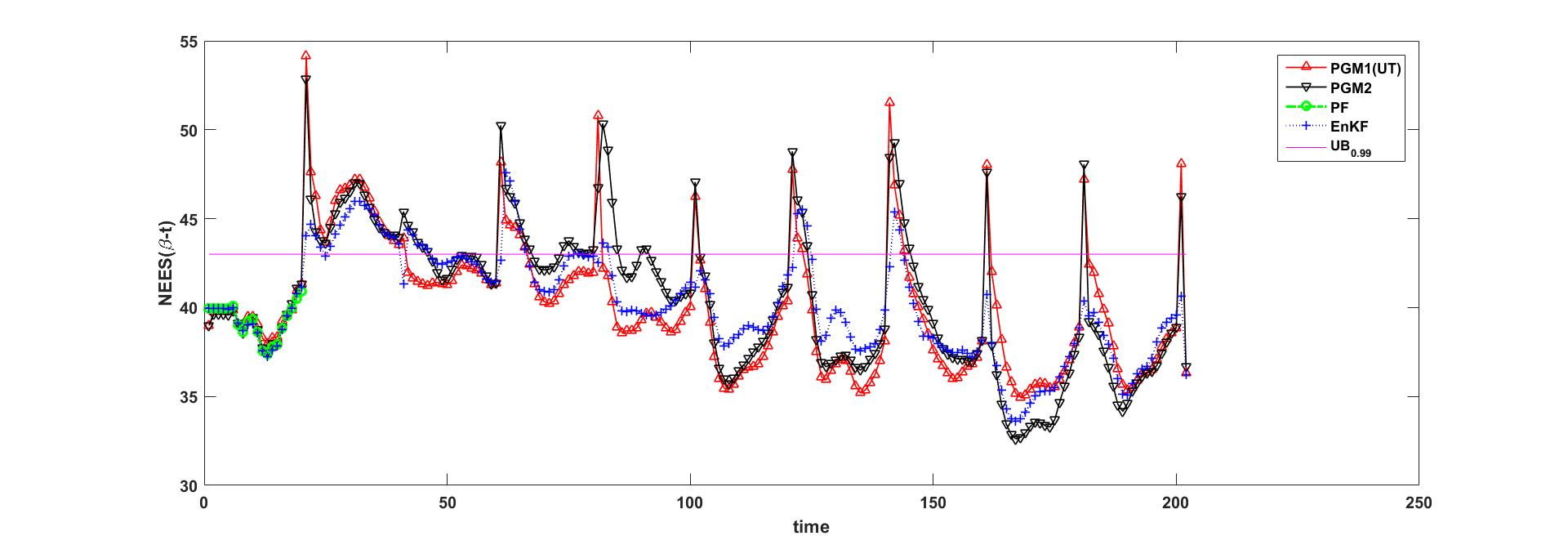}
\caption{NEES($\beta_{t}$)}
\label{fig:ex3nees}
\end{subfigure}

\begin{subfigure}[t]{0.45\textwidth}
\includegraphics[width=\linewidth,height=4cm,]{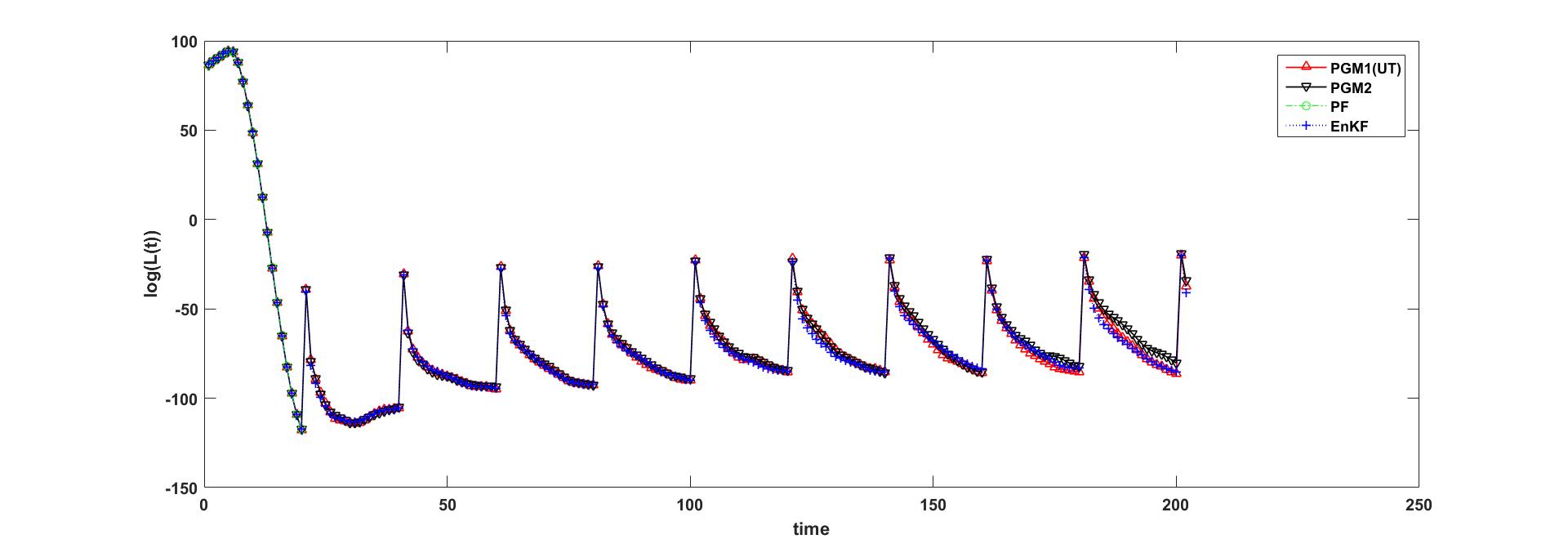}
\caption{$log(L(t))$}
\label{fig:ex3likelihood}
\end{subfigure}
\hspace{\fill}
\begin{subfigure}[t]{0.45\textwidth}
\includegraphics[width=\linewidth,height=4cm,]{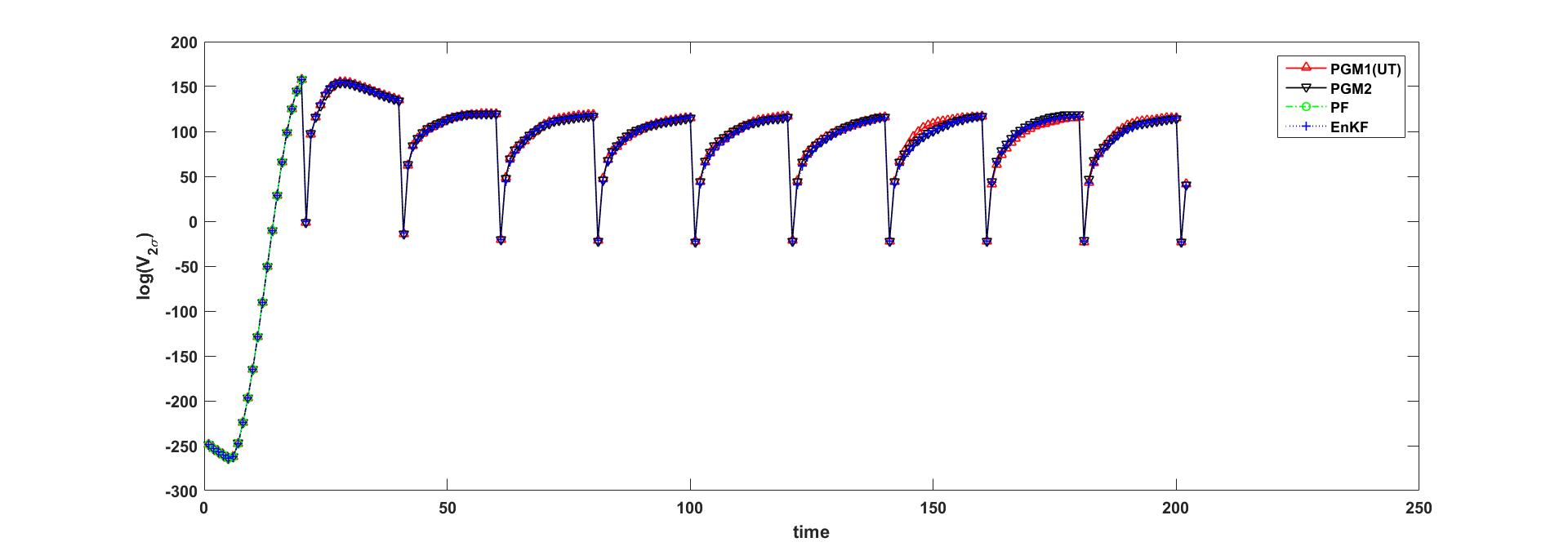}
\caption{$log(V_{2\sigma})$}
\label{fig:ex3volume}
\end{subfigure}
\caption{Example 3 Results} 
\end{figure}
The $99\%$ probability upper bound on $\beta_{t}$ is found to be $Ub_{0.99}=43.0013$. The $y=Ub_{0.99}$ line is included in the NEES plot for reference. The NEES plots show that the EnKF and the PGM1 filter offers comparable performance.  The NEES test results for the PF is seen to blow up after the first measurement update. This clearly indicates that the number of particles employed in the present study is inadequate for a PF based estimation without risking particle depletion.  Finally the logarithm of Monte Carlo averaged likelihoods and the $V\sigma_{2}$ volumes plotted in figures \ref{fig:ex3likelihood},\ref{fig:ex3volume} are used to compare the informativeness of the estimates. In comparison to the EnKF, the PGM filters are seen to perform better in terms of the $v_{2\sigma}$ volume, whereas the EnkF estimates are seen to have the highest averaged likelihoods. The time averaged values of the RMS error $\overline{E_{rms}}$, likelihood ($\hat{L}$) and the $\hat{V}\sigma_{2}$ are listed in table along with the consistency fractions. The results show that the PGM2 is the most accurate estimator of the four filters studied. The EnKF is seen to offer the best performance in terms of consistency fractions , closely followed by the PGM filters. The value of $Sw_{tz}$ was found to stay within the $99\%$ bounds during $60\%$ of the time for both PGM1 and PGM2 filter.

\begin{table}[h!]  
\centering
\caption{}
\label{tab:tab3}
\scalebox{0.8}{
   \begin{tabular}{ |c|c|c|c|c| }
     \hline
  \multicolumn{5}{|c|}{Example 3: Results} \\
  \hline
 & $\overline{E_{rms}}$ & $\beta_{t,c}(\%)$ & $\hat {log L }$ & $ \hat{log V\sigma_{2}}$ \\ 
\hline
  PGM 1(UT) & 18.0069& 80.69  &  89.6553  & 152.8588\\
  \hline
  PGM 2 &   18.0452 & 70.30 &  89.6227  & 152.7732\\
  \hline
  PF &   31.7261 & 9.90 &   89.6244  &  152.7548\\
  \hline
  EnKF  & 18.1055& 81.19  &  89.8193& 152.8034\\
  \hline
\end{tabular}
}
\end{table}

\section{Conclusions}
A novel Gaussian mixture-particle PGM algorithm for nonlinear filtering has been presented. During the prediction step, the PGM filter uses an ensemble of particles to propagate the prior uncertainty. The propagated ensemble is clustered to recover a GMM representation of the propagated pdf. Measurements are incorporated through a Kalman update of the mixture modes to arrive at the posterior pdf. The PGM approach allows the number and weight of mixture components to be adapted during propagation unlike the conventional mixture filters. Additionally, the PGM is not prone to the curse of dimensionality associated with particle measurement updates. The PGM filter density is shown to converge weakly to the true filter density under the condition of exponential forgetting of initial conditions by the true filter. The PGM filter is employed in three test cases to evaluate the etimation performance. It is demonstrated that the PGM filter offers superior estimation performance in comparison to UKF, PF and a mixture UKF. The PGM filter is demonstrated to be capable of tracking the 40 dimensional lorenz 96 system wherein the PF suffers particle depletion. The design of a PGM filtering scheme that incorporates splitting of mixture modes during the measurement update is a goal of future work. Further, the effect of the clustering scheme on the PGM filter needs to be rigorously evaluated.

\bibliographystyle{IEEEtran}
\bibliography{references}

\begin{thebibliography}{10}
\providecommand{\url}[1]{#1}
\csname url@samestyle\endcsname
\providecommand{\newblock}{\relax}
\providecommand{\bibinfo}[2]{#2}
\providecommand{\BIBentrySTDinterwordspacing}{\spaceskip=0pt\relax}
\providecommand{\BIBentryALTinterwordstretchfactor}{4}
\providecommand{\BIBentryALTinterwordspacing}{\spaceskip=\fontdimen2\font plus
\BIBentryALTinterwordstretchfactor\fontdimen3\font minus
  \fontdimen4\font\relax}
\providecommand{\BIBforeignlanguage}[2]{{%
\expandafter\ifx\csname l@#1\endcsname\relax
\typeout{** WARNING: IEEEtran.bst: No hyphenation pattern has been}%
\typeout{** loaded for the language `#1'. Using the pattern for}%
\typeout{** the default language instead.}%
\else
\language=\csname l@#1\endcsname
\fi
#2}}
\providecommand{\BIBdecl}{\relax}
\BIBdecl

\bibitem{Kal}
R.~E. Kalman, ``A new approach to linear filtering and prediction problems,''
  \emph{Transactions of the ASME--Journal of Basic Engineering}, no.~82, pp.
  35--45, 1960.

\bibitem{Bucy}
R.~E. Kalman and R.~S. Bucy, ``New results in linear filtering and prediction
  theory,'' \emph{Transactions of ASME-Journal of Basic Engineering}, vol.~83,
  pp. 96--108, 1961.

\bibitem{smith}
G.~Smith, S.~Schmidt, and L.~McGee, ``Application of statistical filter theory
  to the optimal estimation of position and velocity on board a circumlunar
  vehicle,'' NASA, Tech. Rep. NASA TR-135, Jan. 1962.

\bibitem{Jul}
S.~J. Julier, J.~K. Uhlmann, and H.~Durrant-Whyte, ``A new approach for
  filtering nonlinear systems,'' in \emph{Proceedings of the American Control
  Conference}, 1995, pp. 1628--1632.

\bibitem{Julier}
S.~J. Julier and J.~K. Uhlmann, ``Unscented filtering and nonlinear
  estimation,'' in \emph{Proceedings of the IEEE}, 2004, pp. 401--402.

\bibitem{Wan}
E.~Wan and R.~V.~D. Merwe, ``The unscented kalman filter,'' in \emph{Kalman
  Filter and Neural Networkst}, S.~Haykin, Ed.\hskip 1em plus 0.5em minus
  0.4em\relax New York: J.Wiley and Sons, 2001.

\bibitem{Lefeb}
T.~Lefevbre, H.~Bruyninckx, and J.~D. Schutter, ``Comment on "a new method for
  the nonlinear transformation of means and covariances in filters and
  estimatiors",'' \emph{IEEE Transcations on Automatic Controll}, vol.~47,
  no.~8, pp. 1406--1409, August 2002.

\bibitem{Alspach}
H.~Sorenson and D.~Alspach, ``Recursive bayesian estimation using gaussian
  sums,'' \emph{Automatica}, vol.~7, no.~4, pp. 465--479, 1971.

\bibitem{Sorenson}
D.~Alspach and H.~Sorenson, ``Nonlinear bayesian estimation using gaussian sum
  approximations,'' \emph{IEEE Transactions on Automatic Control}, vol.~17,
  no.~4, pp. 439--448, 1972.

\bibitem{Terejanu}
G.~Terejanu, P.~Singla, T.~Singh, and P.~Scott, ``Adaptive gaussian sum filter
  for nonlinear bayesian estimation,'' \emph{IEEE Transactions on Automatic
  Control}, vol.~56, no.~9, pp. 2151--2156, 2011.

\bibitem{Mars}
K.~DeMars, R.~Bishop, and M.~Jah, ``Entropy-based approach for uncertainty
  propagation of nonlinear dynamical systems,'' \emph{Journal of Guidance,
  Control and Dynamics}, vol.~36, no.~4, pp. 1047--1056, 2013.

\bibitem{Gor}
N.~Gordon, D.~Salmond, and A.~Smith, ``A novel approach to
  nonlinear/non-gaussian bayesian state estimation,'' \emph{IEEE Proceedings F,
  Radar and Signal Processing}, vol. 140, no.~2, pp. 107--113, 1993.

\bibitem{Aru}
S.~Arulampalam, S.~Maskell, N.~Gordon, and T.~Clapp, ``A tutorial on particle
  filters for online nonlinear/non-gaussian bayesian tracking,'' \emph{IEEE
  Transactions on signal processing}, vol.~50, no.~2, pp. 174--188, 2001.

\bibitem{Beng}
T.~Bengtsson, P.~Bickel, and B.Li, ``Curse-of-dimensionality revisited:
  Collapse of particle filter in very large scale systems.'' in
  \emph{Probability and Statistics: Essays in Honor of David A. Freedman},
  2008, vol.~2, pp. 316--334.

\bibitem{Hua}
F.~Daum and J.~Huang, ``Curse of dimensionality and particle filters.'' in
  \emph{Proceedings of IEEE Aerospace Conference}, 2003, pp. 1979--1993.

\bibitem{Geven}
G.~Evenson, \emph{Data Assimilation:The Ensemble Kalman Filter}.\hskip 1em plus
  0.5em minus 0.4em\relax Berlin: Springer, 2002.

\bibitem{FPF}
T.~Yang, P.~G. Mehta, and S.~P.Meyn, ``Feedback particle filter,'' \emph{IEEE
  Transactions on Automatic Control}, vol.~58, no.~10, pp. 2465 --2480, 2013.

\bibitem{UKPF}
A.~V.~D. Raihan and S.~Chakravorty, ``A ukf-pf based hybrid estimation scheme
  for space object tracking,'' in \emph{Proceedings of the AAS/AIAA
  astrodynamics specialist conference}, 2015, to appear.

\bibitem{Duda}
Richard.O.Duda, P.~E.Hart, and D.~G.Stork, \emph{Pattern Classification},
  2nd~ed.\hskip 1em plus 0.5em minus 0.4em\relax New York: Wiley-Interscience,
  November 2000.

\bibitem{Jain1}
A.K.Jain, M.N.Murthy, and P.J.Flynn, ``Data clustering: a review,'' \emph{ACM
  Computing Surveys}, vol.~13, no.~3, pp. 264--323, 1999.

\bibitem{Macq}
J.Macqueen, ``Some methods for classification and analysis of multivariate
  observations,'' 1967, pp. 281--297.

\bibitem{Demp}
A.~P. Dempster, N.~M. Laird, and D.~Rubin, ``Maximum likelihood from incomplete
  data via the em algorithm,'' \emph{Journal of the Royal Statistical Society,
  Series B}, vol.~39, no.~1, pp. 1--38, 1977.

\bibitem{sugar}
C.~A.Sugar, ``Techniques for clustering and classification with applications to
  medical problems,'' Ph.D. dissertation, Department of Staistics, Stanford
  University, Stanford,CA, 1998.

\bibitem{tibsh}
R.~Tibshirani, G.~Walther, and T.~Hastie, ``Estimating the number of clusters
  ina dataset via the gap statistic,'' \emph{Journal of the Royal Statistical
  Society, Series B}, vol.~63, no.~2, pp. 411--423, 2001.

\bibitem{figjain}
M.~Figueiredo and A.~Jain, ``Unsupervised learning of finite mixture models,''
  \emph{IEEE transactions on Pattern Analysis and Machine Intelligence},
  vol.~24, no.~3, pp. 381--396, 2002.

\bibitem{beng2}
C.~Snyder, T.~Bengtsson, P.~Bickel, and J.~Anderson, ``Obstacles to high-
  dimensional particle filtering,'' \emph{Monthly Weather Review}, vol. 136,
  no.~12, pp. 4629--4640, 2008.

\bibitem{Beng3}
P.~Bickel, B.~Li, and T.~Bengtsson, ``Sharp failure rates for the bootstrap
  filter in high dimnesions,'' in \emph{IMS Collections:Pushing the Limits of
  Contemporary Statistics: Contributions in Honor of Jayanta K Ghosh}, 2008,
  vol.~3, pp. 318--329.

\bibitem{Eveng}
G.Evensen, ``Sequential data assimilation with a nonlinear quasigeostrophic
  model using monte carlo methods to forecast error statistics,'' \emph{Journal
  of Geophysical Research: Oceans}, vol.~99, no.~C5, pp. 10 143--10 162, 1994.

\bibitem{Burgers}
G.~Burgers, P.~V. Leeuwen, and G.Evensen, ``Analysis scheme in the ensemble
  kalman filter,'' \emph{Monthyly Weather Review}, vol. 126, no.~6, pp.
  1719--1724, 1998.

\bibitem{Ander}
J.~L. Anderson, ``Ensemble kalman filters for large geophysical applications,''
  \emph{IEEE control systems magazine}, vol.~29, pp. 66--82, 2009.

\bibitem{lloyd}
S.~P. Lloyd, ``Least squares quantization in pcm,'' \emph{IEEE Transactions on
  Information Theory}, vol.~28, no.~2, pp. 129--137, 1982.

\bibitem{Pell}
D.~Pelleg and A.~Moore, ``Accelerating exact k-means algorithms with geometric
  reasoning,'' in \emph{KDD '99 Proceedings of the fifth ACM SIGKDD
  international conference on Knowledge discovery and data mining}, 1999, pp.
  277--281.

\bibitem{Beck}
U.~D. Hanebeck, K.~Briechle, and A.~Rauh, ``Progressive bayes: a new framework
  for nonlinear state estimation,'' in \emph{SPIE vol.5099 Multisensor,
  Multisource Information Fusion: Architectures, Algorithms, and Applications},
  2003, pp. 256--267.

\bibitem{Bail}
T.~Bailey, J.~Nieto, J.~Guivant, M.Stevens, and E.Nebot, ``Consistency of the
  ekf-slam algorithm,'' in \emph{Proceedings of the 2006 IEEE/RSJ International
  Conference on Intelligent Robots and Systems}, 2006, pp. 3562--3568.

\bibitem{Lef2}
T.~Lefebvre, H.~Bruyninckx, and J.~D. Schutter, ``Kalman filters for non-linear
  systems: a comparison of performance,'' \emph{International Journal of
  Control}, vol.~77, no.~7, pp. 639--653, 2004.

\end{thebibliography}


\end{document}